\documentclass[11pt,letterpaper,oneside]{article}
\usepackage[letterpaper,margin=1in]{geometry}
\usepackage[T1]{fontenc}
\usepackage{lmodern}
\usepackage{graphicx}
\usepackage{amsmath}
\usepackage{amssymb}
\usepackage{amsthm}
\usepackage[scr=pxtx]{mathalpha}
\usepackage{mathtools}
\usepackage[colorlinks=true, citecolor=blue, urlcolor=blue, pagebackref=true, linktocpage, plainpages=false, pdfpagelabels]{hyperref}

\usepackage{xcolor}
\usepackage{cancel}
\usepackage{subcaption}
\usepackage{paralist}
\usepackage{booktabs}
\usepackage{placeins}
\usepackage{float}
\usepackage{enumitem}
\usepackage{cleveref}
\usepackage{tikz-cd}
\usepackage{tikz}
\usepackage{pgfplots}
\pgfplotsset{compat=1.17}

\usepackage[textwidth=2.5cm]{todonotes}

\newtheorem{theorem}{Theorem}[section]
\numberwithin{equation}{section}
\newtheorem{lemma}[theorem]{Lemma}

\newtheorem{corollary}[theorem]{Corollary}

\theoremstyle{definition}
\newtheorem{definition}[theorem]{Definition}

\renewcommand{\setminus}{\smallsetminus}

\newcommand{\bfa}{\mathbf{a}}
\newcommand{\bfb}{\mathbf{b}}
\newcommand{\bfc}{\mathbf{c}}
\newcommand{\bfu}{\mathbf{u}}
\newcommand{\bfx}{\mathbf{x}}
\newcommand{\bfy}{\mathbf{y}}
\newcommand{\bfz}{\mathbf{z}}
\newcommand{\bfL}{\mathbf{L}}
\newcommand{\bfR}{\mathbf{R}}
\newcommand{\bfzero}{\mathbf{0}}
\newcommand{\bfone}{\mathbf{1}}

\newcommand{\sfk}{\mathsf{k}}
\newcommand{\sfL}{\mathsf{L}}
\newcommand{\sfR}{\mathsf{R}}

\newcommand{\PAR}{\mathsf{PAR}}
\newcommand{\OR}{\mathsf{OR}}
\newcommand{\EXACTZ}{\mathsf{EXACT}_{\mathbb{Z}}}
\newcommand{\SUPZ}{\mathsf{SUP}_{\mathbb{Z}}}

\newcommand{\calA}{\mathcal{A}}
\newcommand{\calB}{\mathcal{B}}
\newcommand{\calF}{\mathcal{F}}
\newcommand{\calH}{\mathcal{H}}
\newcommand{\calI}{\mathcal{I}}
\newcommand{\calK}{\mathcal{K}}

\newcommand{\calY}{\mathcal{Y}}
\newcommand{\calZ}{\mathcal{Z}}

\DeclareMathOperator{\h}{h}
\DeclareMathOperator{\Mid}{Mid}
\DeclareMathOperator{\nnz}{nnz}

\definecolor{diffdelete}{RGB}{180,35,35}
\definecolor{diffinsert}{RGB}{20,85,170}

\newcommand{\email}[1]{\href{mailto:#1}{\texttt{#1}}}

\title{Subset Sum via Partial Match}
\author{Baitian Li\thanks{Columbia University. \email{bl3052@columbia.edu}. Supported in part by NSF Grant CCF-2238221 and a Packard
Foundation Fellowship.} \and Lixi Ye\thanks{Institute for Interdisciplinary Information Sciences, Tsinghua University. \email{ylx25@mails.tsinghua.edu.cn}.}}
\date{}

\usepackage{array,longtable,calc}
\usepackage{xurl}

\hypersetup{pdftitle={Subset Sum via Partial Match},
  pdfauthor={Baitian Li, Lixi Ye}}

\begin{document}
\maketitle

\begin{abstract}
  We show that the subset sum problem can be solved in randomized time $O(2^{0.49999\overline{7}n})$, breaking the $2^{n/2}$ meet-in-the-middle barrier.

  Our approach builds on the representation technique framework of Randolph and W\k{e}grzycki (STOC 2026). We choose a representation for which testing the compatibility of pairs of partial solution vectors is exactly the \emph{partial match} problem. To beat exponent $1/2$ for subset sum, it then suffices to give a nontrivial partial match algorithm in a certain parameter regime. We achieve this by designing a depth-2 linear circuit for the partial match matrix, which yields an efficient algorithm via the framework of Alman and Li (FOCS 2025).
\end{abstract}

\paragraph*{Declaration.}

A result of a similar type was recently released by OpenAI: \emph{\href{https://github.com/openai/math/blob/main/preprints/Subset-Sum-in-Time-2-power-0-49n-October-4-2026/subset-sum.pdf}{Subset
Sum in Time $O(2^{0.49n})$}} \cite{OpenAI26}, dated 4 October 2026. Our result was developed independently of that work and uses a different approach; see \Cref{sec:comparison} for a brief comparison of the two approaches and \hyperref[sec:ai-disclosure]{AI disclosure and timestamps} at the end of the paper for details.

\clearpage
\section{Introduction}\label{sec:introduction}

Given $n$ integers and a target, the \emph{subset sum} problem asks whether some subset of the integers sums to the target. Formally, given a vector $\bfz\in\mathbb{Z}^n$ and a target $\boldsymbol{\tau}\in\mathbb{Z}$, it asks for a vector $\bfc\in\{0,1\}^n$ such that $\bfc\cdot\bfz=\boldsymbol{\tau}$.

In their foundational work, Horowitz and Sahni~\cite{HS74} gave an $O^*(2^{n/2})$-time algorithm for subset sum,\footnote{We use $O^*(\cdot)$ to suppress polynomial factors in $n$ and the input length (measured in bits).} a method now known as \emph{meet-in-the-middle}. Chen, Jin, Randolph, and Servedio~\cite{CJRS23} later gave a randomized $O(2^{n/2}n^{-\gamma})$-time algorithm with $\gamma\approx0.5023$, in the word RAM model. Whether subset sum can be solved \emph{exponentially} faster, say in time $O(2^{(1/2-\varepsilon)n})$ for some constant $\varepsilon>0$, has been a longstanding open question in exact algorithms.

Howgrave-Graham and Joux~\cite{HJ10} were the first to break the meet-in-the-middle barrier in an \emph{average-case} setting of subset sum, with running time $O(2^{0.337n})$. Their algorithm was subsequently improved in~\cite{BCJ11,BBSS20}.
Although these algorithms are not worst-case algorithms, \cite{HJ10} introduced the \emph{representation technique}, which later proved very fruitful. Mucha, Nederlof, Pawlewicz, and W\k{e}grzycki~\cite{MNPW19} showed that the representation technique also yields a faster worst-case algorithm for a related problem, \emph{equal subset sum}, in which the coefficient vector ranges over $\bfc\in\{-1,0,1\}^n\setminus\{\bfzero\}$ and the target is $0$. Compared with the meet-in-the-middle bound $O^*(3^{n/2})\subseteq O(1.7321^n)$, they achieved time $O(1.7088^n)$.

Recently, Randolph and W\k{e}grzycki~\cite{RW26} systematically studied a class of problems related to subset sum, which we call \emph{generalized subset sum} (GSS). Here the target is $0$, and the coefficient vector ranges over $C^n$ or $C^n\setminus\{\bfzero\}$ for a fixed finite set $C\subseteq\mathbb{Z}$ that defines the problem. Their starting point is the earlier work of Nederlof and W\k{e}grzycki~\cite{NW21}, which used an algorithm for detecting sparse Orthogonal Vectors to improve the space complexity of the $O^*(2^{n/2})$-time, $O^*(2^{n/4})$-space subset sum algorithm of Schroeppel and Shamir~\cite{SS81} to $O^*(2^{(1/4-\varepsilon)n})$, while keeping the time $O^*(2^{n/2})$. Randolph and W\k{e}grzycki observed that this Orthogonal Vectors subroutine generalizes to a wider family of problems, which they call \emph{compatibility testing}. Combining compatibility testing with the representation technique, they showed that GSS can be solved in time $O(\lvert C\rvert^{(1/2-\varepsilon)n})$ for a large family of sets $C$, beating meet-in-the-middle. In particular, they beat meet-in-the-middle for $C=\{-d,\ldots,d\}$ for every $d>1$, and for $C=[\pm d]\coloneqq\{-d,\ldots,-1,1,\ldots,d\}$ for every $d>2$. They also improved the equal subset sum bound of~\cite{MNPW19} to $O^*(1.7067^n)$. For $C=[\pm1]$, GSS is equivalent to subset sum via the affine change of coefficients $\bfc\mapsto2\bfc-\bfone$. This leaves the tantalizing possibility that their method extends to subset sum itself.

In this paper, we show that this line of approach indeed solves subset sum in time below $2^{n/2}$.

\begin{theorem}[{Theorem~\ref{thm:main}}]
  Subset sum can be solved in randomized time $O(2^{0.49999\overline{7}n})$.\footnote{Here $0.49999\overline{7}=0.49999777\ldots=1/2-1/450000$.}
\end{theorem}

\subsection{The representation technique}

We first describe the framework of the representation technique~\cite{HJ10,MNPW19,RW26}, on which our algorithm builds.

The starting point of the representation technique is to carefully choose sets $\calA,\calB\subseteq\mathbb{Z}^n$ and then search for $\bfa\in \calA$ and $\bfb\in \calB$ such that $\bfc=\bfa+\bfb$ is a solution. For example, the classical meet-in-the-middle algorithm takes
\[
 \calA=\{0,1\}^{n/2}\times0^{n/2},\qquad
 \calB=0^{n/2}\times\{0,1\}^{n/2}.
\]

At first glance, this approach cannot break the $2^{n/2}$ barrier: to cover all $2^n$ candidate solutions, we need $\lvert\calA\rvert\cdot\lvert\calB\rvert\ge2^n$, so at least one of $\calA$ and $\calB$ has size at least $2^{n/2}$. The key observation of the representation technique~\cite{HJ10}, as used in~\cite{MNPW19}, is that one need not list $\bfa\cdot\bfz$ and $\bfb\cdot\bfz$ for all $\bfa\in \calA$ and $\bfb\in \calB$. The idea is that, instead of making $\calA$ and $\calB$ as small as possible, one should first enlarge them: the benefit is that each solution $\bfc$ then admits many representations $\bfc=\bfa+\bfb$ with $\bfa\in \calA$ and $\bfb\in \calB$.

\paragraph{Mixing and hashing.}
Suppose that the choice of $\calA$ and $\calB$ guarantees that every solution $\bfc$ has at least $N_{\mathrm{rep}}$ representations $(\bfa,\bfb)\in \calA\times \calB$. Following~\cite{AKKN16,MNPW19,RW26}, one first analyzes the algorithm under an idealized \emph{mixing} assumption, and then complements this analysis with a \emph{mixing dichotomy}: whenever the mixing assumption fails, the instance can be solved by a different, faster algorithm. The algorithm of~\cite{RW26}, as well as ours, relies on a more technically involved form of this dichotomy. Nonetheless, throughout this overview the reader may keep in mind the simplest form of the mixing assumption, which is the easiest to analyze: the sums $\bfa\cdot\bfz$ are \emph{pairwise distinct} over all representations $\bfa+\bfb=\bfc$.

In this case, one can avoid listing all $\lvert\calA\rvert+\lvert\calB\rvert$ sums as follows. Pick a random prime $p$ of size near $N_{\mathrm{rep}}$, and then a uniformly random residue $r$ modulo $p$.
\begin{enumerate}
 \item List all $\bfa\in \calA$ with $\bfa\cdot\bfz\equiv r\pmod p$ and all $\bfb\in \calB$ with $\bfb\cdot\bfz\equiv\boldsymbol{\tau}-r\pmod p$; denote these sets by $\boldsymbol{A}$ and $\boldsymbol{B}$, respectively.
 \item Find $\bfa\in \boldsymbol{A}$ and $\bfb\in \boldsymbol{B}$ such that $\bfa+\bfb$ is a solution.\label{intro:item:sift_list}
\end{enumerate}
With suitable hashing parameters and repetitions, a representation survives with high probability. In each trial, the expected filtered list sizes are roughly $\lvert\calA\rvert/N_{\mathrm{rep}}$ and $\lvert\calB\rvert/N_{\mathrm{rep}}$, up to polynomial factors.

For the choices of $\calA$ and $\calB$ below, $\boldsymbol{A}$ and $\boldsymbol{B}$ can be listed in time roughly $\lvert\boldsymbol{A}\rvert+\lvert\boldsymbol{B}\rvert$, together with the cost of generating suitable partial lists, so this part can run well below $2^{n/2}$. However, another technical problem arises: if $\calA+\calB\not\subseteq\{0,1\}^n$, then step~\ref{intro:item:sift_list} becomes harder, since one must sift out the \emph{pseudosolutions}, i.e., pairs with $\bfa\cdot\bfz+\bfb\cdot\bfz=\boldsymbol{\tau}$ but $\bfa+\bfb\notin\{0,1\}^n$.

The insight of~\cite{RW26} is that a carefully optimized algorithm for detecting true solutions can be used to speed up certain generalized subset sum problems. As an illustrative application, consider the case where the coefficient vector $\bfc$ may range over $\{-1,0,1\}^n$. In~\cite{MNPW19}, the authors showed how to beat the meet-in-the-middle bound $3^{n/2}$ by taking $\calA=\{0,1\}^n$ and $\calB=\{0,-1\}^n$. Then~\cite{RW26} improved this exponent by taking $\calA$ and $\calB$ to be carefully chosen subsets of larger sets,
\[
 \calA\subseteq\{0,1,2\}^n,\qquad \calB\subseteq\{-2,-1,0\}^n.
\]
The values of $a_i+b_i$ are given by the following table:
\[
\begin{array}{c|ccc}
 a_i\backslash b_i&-2&-1&0\\\hline
 0&-2&\mathbf{-1}&\mathbf{0}\\
 1&\mathbf{-1}&\mathbf{0}&\mathbf{1}\\
 2&\mathbf{0}&\mathbf{1}&2
\end{array}
\]
where the entries in bold are those with $a_i+b_i\in\{-1,0,1\}$. Since a true solution must lie in $\{-1,0,1\}^n$, the relevant compatibility matrix is
\[
P_{\mathrm{RW}}=\begin{array}{c|ccc}
 a_i\backslash b_i&-2&-1&0\\\hline
 0&0&1&1\\
 1&1&1&1\\
 2&1&1&0
\end{array}.
\]
Sifting out pseudosolutions in step~\ref{intro:item:sift_list} thus leads to the following algorithmic problem: given sets of vectors $\boldsymbol{A}\subseteq\{0,1,2\}^n$ and $\boldsymbol{B}\subseteq\{-2,-1,0\}^n$, find $\bfa\in \boldsymbol{A}$ and $\bfb\in \boldsymbol{B}$ such that the corresponding entry of $P_{\mathrm{RW}}^{\otimes n}$ is $1$, i.e., $(P_{\mathrm{RW}})_{a_i,b_i}=1$ for every $i$. Randolph and W\k{e}grzycki~\cite{RW26} call this problem \emph{compatibility testing}; they designed a fast compatibility testing algorithm for $P_{\mathrm{RW}}$, which improves the exponent obtained in~\cite{MNPW19}. One of our key technical contributions is a different choice of representation, whose compatibility matrix admits a surprisingly fast testing algorithm.

\subsection{Algorithm overview}\label{sec:overview}

We now describe our algorithm at a level of detail sufficient to verify its running time; Sections~\ref{sec:reduction}--\ref{sec:together} give the formal proof. Throughout, $\h(t)\coloneqq-t\log_2t-(1-t)\log_2(1-t)$ denotes the binary entropy function, with $0\log_20\coloneqq0$.

\paragraph{The representation and its compatibility matrix.}
We represent $\bfc\in\{0,1\}^m$ as $\bfc=\bfa+\bfb$ with $\bfa\in\{0,1\}^m$ and $\bfb\in\{-1,0,1\}^m$, using
\[
 0=0+0=1+(-1),\qquad 1=0+1=1+0.
\]
For $s\in[0,1]$, let $T_{m,s}\subseteq\{-1,0,1\}^m$ be the set of vectors with exactly $sm$ nonzero entries. The values of $a_i+b_i$ are given by the following table:
\[
\begin{array}{c|ccc}
 a_i\backslash b_i&-1&0&1\\\hline
 0&-1&\mathbf{0}&\mathbf{1}\\
 1&\mathbf{0}&\mathbf{1}&2
\end{array}
\]
where the entries in bold are those with $a_i+b_i\in\{0,1\}$. Accordingly, the compatibility matrix is
\[
P=\begin{array}{c|ccc}
 a_i\backslash b_i&-1&0&1\\\hline
 0&0&1&1\\
 1&1&1&0
\end{array},
\]
and compatibility testing asks, given $\boldsymbol{A}\subseteq\{0,1\}^m$ and $\boldsymbol{B}\subseteq T_{m,s}$, whether some $\bfa\in \boldsymbol{A}$ and $\bfb\in \boldsymbol{B}$ satisfy $\bfa+\bfb\in\{0,1\}^m$, i.e., whether the corresponding entry of $P^{\otimes m}$ is $1$.

Relabeling the columns via $1\mapsto\mathtt0$, $-1\mapsto\mathtt1$, and $0\mapsto\mathtt*$, and ordering them as $\mathtt0,\mathtt1,\mathtt*$, turns $P$ into the partial match matrix
\[
P_1=\begin{array}{c|ccc}
 &\mathtt0&\mathtt1&\mathtt*\\\hline
 \mathtt0&1&0&1\\
 \mathtt1&0&1&1
\end{array},
\]
where $\mathtt*$ is a wildcard matching both $\mathtt0$ and $\mathtt1$. We call strings over $\{0,1\}$ \emph{binary strings} and strings over $\{0,1,*\}$ \emph{query strings} (or queries). Under this relabeling, $T_{m,s}$ becomes the set of queries with exactly $(1-s)m$ wildcards. Thus compatibility testing for our representation is exactly partial match. A key technical ingredient of our algorithm is the following lemma for the ``half-stars'' case $s=1/2$.

\begin{lemma}\label{intro:lem:half-stars}
 There are nonnegative constants $l,r$ and a constant $\xi>0$ with $l+r\le1/2-2\xi$ such that, for every $\varepsilon>0$, there is a deterministic algorithm which, for even $m$, tests compatibility for $s=1/2$ in time
 \[
 O\!\left(2^{(l+\varepsilon)m}\lvert\boldsymbol{A}\rvert+2^{(r+\varepsilon)m}\lvert\boldsymbol{B}\rvert\right).
 \]
\end{lemma}

Other support densities $s<1/2$ reduce to the case $s=1/2$ by a simple padding argument: we append $(1-2s)m$ dummy pairs $(1,-1)$, which are always compatible, to every left and right vector. After relabeling, the padded strings have even length $(2-2s)m$, and exactly half of the positions of each query are specified. Lemma~\ref{intro:lem:half-stars} then tests compatibility in time
\begin{equation}\label{eqn:padding}
  O\!\left(2^{2(1-s)(l+\varepsilon)m}\lvert\boldsymbol{A}\rvert+2^{2(1-s)(r+\varepsilon)m}\lvert\boldsymbol{B}\rvert\right).
\end{equation}

If ``$\xi>0$'' is relaxed to ``$\xi\ge0$'', the lemma would be trivial: the pair $(l,r)=(0,1/2)$ is achieved by letting each query enumerate all $2^{m/2}$ completions of its wildcards. As we show below, beating this trivial bound $l+r=1/2$ by any positive constant already yields a faster subset sum algorithm.

\paragraph{The search space.}
We use representations only on the first $n/3$ coordinates, which we call the \emph{prefix}, and plain meet-in-the-middle on the remaining $2n/3$; this lets us balance the asymmetric costs of the two sides. For a parameter $\ell$ chosen below and a support density $s\in[0,1/2]$, let
\[
 \calA=\{0,1\}^{n/3}\times\{0,1\}^{\ell}\times0^{2n/3-\ell},\qquad
 \calB=T_{n/3,s}\times0^{\ell}\times\{0,1\}^{2n/3-\ell},
\]
so that
\[
 \lvert\calA\rvert=2^{n/3+\ell},\qquad
 \lvert\calB\rvert=\binom{n/3}{sn/3}2^{sn/3+2n/3-\ell}.
\]
We run the algorithm for every possible value of $s$.

\paragraph{Mixing case: many correct keys.}
Fix a solution $\bfc$. The representations $\bfc=\bfa+\bfb$ with $(\bfa,\bfb)\in\calA\times\calB$ differ only on the prefix, and each binary prefix of $\bfa$ determines $\bfb=\bfc-\bfa$. Let $N_{\mathrm{sum}}$ be the number of distinct subset sums of the first $n/3$ input numbers. The mixing assumption discussed above corresponds to the perfect case $N_{\mathrm{sum}}=2^{n/3}$, in which all representations of $\bfc$ have distinct keys $\bfa\cdot\bfz$. We replace it by the weaker assumption that $N_{\mathrm{sum}}$ is large, and determine below how large it needs to be.

We first observe that some support layer contains many distinct keys. Pair each binary prefix of $\bfa$ with its complement $\bfa'$. The corresponding vectors $\bfb=\bfc-\bfa$ and $\bfb'=\bfc-\bfa'$ have prefix support sizes $k$ and $n/3-k$, respectively, and the prefix sums of $\bfa$ and $\bfa'$ add up to the total prefix sum, which is fixed. Hence at least $N_{\mathrm{sum}}/2$ distinct prefix sums have a representation with prefix support size at most $n/6$, so some layer $s\le1/2$ contains $\Omega(N_{\mathrm{sum}}/n)$ distinct correct keys.

Since we run over all $s$, it suffices for the algorithm to succeed in this layer. We hash the keys modulo a random prime of size near $N_{\mathrm{sum}}$, as described above.\footnote{We remark that the formal proof uses a slightly different hash function, which avoids some technicalities arising when the input integers have up to $2^{o(n)}$ bits.} The surviving lists $\boldsymbol{A},\boldsymbol{B}$ then have expected sizes roughly
\[
 \mathbb{E}[\lvert\boldsymbol{A}\rvert]=O\!\left(\frac{\lvert\calA\rvert}{N_{\mathrm{sum}}}\right),\qquad
 \mathbb{E}[\lvert\boldsymbol{B}\rvert]=O\!\left(\frac{\lvert\calB\rvert}{N_{\mathrm{sum}}}\right),
\]
and they can be generated from meet-in-the-middle partial lists in time close to their sizes.

We then run the padded compatibility testing algorithm~\eqref{eqn:padding} with $m=n/3$ on the prefixes of the surviving vectors, grouping them by their complete integer key so that compatibility and the target-sum condition are checked together. Using $\binom{n/3}{sn/3}\le2^{\h(s)n/3}$, for any fixed $\varepsilon>0$, the expected cost is, up to subexponential factors,
\[
\begin{aligned}
 &\frac{2^{2(1-s)(l+\varepsilon)n/3}\lvert\calA\rvert+2^{2(1-s)(r+\varepsilon)n/3}\lvert\calB\rvert}{N_{\mathrm{sum}}}\\
 &\qquad\le\frac{2^{2(1-s)(l+\varepsilon)n/3+n/3+\ell}+2^{2(1-s)(r+\varepsilon)n/3+(\h(s)+s)n/3+2n/3-\ell}}{N_{\mathrm{sum}}}.
\end{aligned}
\]
Choosing
\[
 \ell\coloneqq\frac{n}{6}\left(1+\h(s)+s-2(l-r)(1-s)\right)
\]
balances the two terms, and their common exponent is
\[
 \frac{n}{2}+\left((l+r+2\varepsilon)(1-s)+\frac{\h(s)+s}{2}\right)\frac{n}{3}-\log_2N_{\mathrm{sum}}.
\]
Using $\h(s)\le1$, $l+r\le1/2-2\xi$, and $s\le1/2$ gives the upper bound
\[
 \frac{n}{2}+\frac{(1-\xi)n}{3}-\log_2N_{\mathrm{sum}}+\frac{2\varepsilon n}{3}.
\]
Consequently, whenever $N_{\mathrm{sum}}\ge2^{(1-\xi+\eta)n/3}$ for some constant $\eta>0$, choosing $\varepsilon<\eta/2$ beats $2^{n/2}$ by an exponential factor.

\paragraph{Mixing dichotomy.}
When $N_{\mathrm{sum}}$ is small, we instead use the few-sums meet-in-the-middle procedure of~\cite{AKKN16}, which runs in time $\sqrt{N_{\mathrm{sum}}}2^{n/3}$. For a threshold $\gamma$,
\begin{itemize}
 \item if $N_{\mathrm{sum}}\le2^{(1-\gamma)n/3}$, the procedure of~\cite{AKKN16} takes time at most $2^{(1/2-\gamma/6)n+o(n)}$;
 \item otherwise, the many-sums procedure above takes time at most $2^{(1/2-(\xi-\gamma)/3+2\varepsilon/3)n+o(n)}$.
\end{itemize}
Choosing $\gamma\coloneqq2\xi/3$ makes both exponents before the added slack $1/2-\xi/9$. Accounting for the partial match algorithm's exponent slack, subset sum can be solved in time $O(2^{(1/2-\xi/9+\varepsilon)n})$ for every $\varepsilon>0$. Our circuit gives $\xi>1/50000$, leaving positive slack before the exponent $1/2-1/450000$ of Theorem~\ref{thm:main}.

\subsection{Partial match via depth-2 linear circuits}\label{sec:overview-circuits}
Finally, we briefly explain how we obtain the partial match algorithm of Lemma~\ref{intro:lem:half-stars}.

\paragraph{Depth-2 linear circuits.}
Let $M\in\{0,1\}^{I\times J}$ be a large, implicitly given $01$ matrix. Given $\boldsymbol{A}\subseteq I$ and $\boldsymbol{B}\subseteq J$, with sizes much smaller than $\lvert I\rvert$ and $\lvert J\rvert$, respectively, we want to decide whether some $(a,b)\in \boldsymbol{A}\times \boldsymbol{B}$ has $M_{a,b}=1$; this is compatibility testing for $M$.

A depth-2 linear circuit for $M$ is simply a matrix factorization $M=\bfL\bfR$. It can be viewed as a circuit computing $\bfu\mapsto\bfu\bfL\bfR$: the input gates carry the coordinates of $\bfu$, each middle gate computes a linear combination of the values at the input gates, and each output gate computes a linear combination of the values at the middle gates. The relevant measure here is the pair of \emph{degrees}: the circuit has degrees $(\deg_{\sfL},\deg_{\sfR})$ if every row of $\bfL$ has at most $\deg_{\sfL}$ nonzero entries and every column of $\bfR$ has at most $\deg_{\sfR}$ nonzero entries.

Given such a circuit, a simple algorithm solves compatibility testing in time $O(\lvert\boldsymbol{A}\rvert\deg_{\sfL}+\allowbreak\lvert\boldsymbol{B}\rvert\deg_{\sfR})$, up to the cost of arithmetic. Writing $M_{a,b}=\bfL_{a,\cdot}\cdot\bfR_{\cdot,b}$, we obtain
\[
 \sum_{a\in \boldsymbol{A},\,b\in \boldsymbol{B}}M_{a,b}=\Bigl(\sum_{a\in \boldsymbol{A}}\bfL_{a,\cdot}\Bigr)\cdot\Bigl(\sum_{b\in \boldsymbol{B}}\bfR_{\cdot,b}\Bigr),
\]
and by the degree bounds, the two vectors on the right have at most $\lvert\boldsymbol{A}\rvert\deg_{\sfL}$ and $\lvert\boldsymbol{B}\rvert\deg_{\sfR}$ nonzero entries, respectively. When the factorization holds over the integers, the left-hand side is exactly the number of compatible pairs; we call such circuits $\EXACTZ$ circuits.

Many known compatibility testing algorithms implicitly construct such a circuit, with the caveat that their arithmetic is over the Boolean semiring rather than the integers; we call these $\OR$ circuits. For example, the Orthogonal Vectors algorithm of~\cite{NW21} covers the Orthogonal Vectors matrix (also known as the disjointness matrix) by combinatorial rectangles, each of which is a rank-one $01$ matrix represented by a single middle gate.

\paragraph{Uniformization.}
Proving the mere \emph{existence} of a good depth-2 linear circuit would a priori yield only a \emph{non-uniform} algorithm, since a uniform algorithm must also \emph{efficiently construct} the circuit. For depth-2 linear circuits, this is not a genuine obstacle, thanks to the \emph{tensor product}: it lets us start from a depth-2 linear circuit of \emph{fixed} size and bootstrap it into an explicit family that handles arbitrarily large $n$.

This phenomenon is sometimes called \emph{uniformization}. Williams~\cite[Theorem~8]{Wil24} showed that a notion called \emph{equality rank} admits uniformization, and combined this with a fixed-size construction for the disjointness matrix to obtain Orthogonal Vectors algorithms. Perhaps the most famous instance of uniformization is matrix multiplication: Strassen~\cite{Str69} found a bilinear algorithm of rank $7$ for $2\times2$ matrix multiplication, and applying it recursively turns this single finite construction into a uniform algorithm for all dimensions.

Alman and Li~\cite{AL25} observed that an equality rank bound implies a degree bound for depth-2 linear circuits. Their framework, which turns a depth-2 linear circuit for a matrix into an algorithm, thus generalizes Williams's approach, and it is the framework we build on.

\paragraph{Rebalancing.}
Uniformization therefore lets us concentrate on bounding the degrees of circuits; explicit constructions, and hence algorithms, follow automatically. Bounding the degrees directly is nevertheless difficult, so we take a two-step route:
\begin{enumerate}[label=(\arabic*)]
  \item construct a family of circuits with complementary advantages, none of which need meet our degree requirement on its own; and
  \item combine their advantages to obtain a genuine improvement in the degrees.
\end{enumerate}
For step~(2), previous works on depth-2 linear circuits~\cite{JS13,AGP23,Ser22,AL25} developed a technique called \emph{rebalancing}. These works use rebalancing to control the circuit \emph{size}, whereas we need to control the \emph{degrees}. Fortunately, \cite{AL25} gives a general analysis of a certain rebalancing process, which adapts to our setting after a slight technical modification.

\paragraph{The partial match matrix.}
For our algorithmic purposes, we need to uniformly construct circuits for the matrix $M=P^{\otimes n}$ with its columns restricted to $T_{n,s}$. After the relabeling above, this is the matrix $P_{n,s}$, which we call a \emph{partial match matrix}.

The rest of the construction is technically involved. In the remainder of this overview, we only sketch the main ideas and their relation to previous work, deferring the details to Section~\ref{sec:circuits}.

A minor technical issue is that $P_{n,s}$ is not itself a tensor power, only a submatrix of the tensor power $P^{\otimes n}$, so the general uniformization results for tensor powers do not apply directly. We overcome this by noting that $P_{n,1/2}$ is still \emph{nearly} a tensor power: up to coordinate permutations, it can be covered by a small number of copies of a tensor power of a fixed-size partial match matrix $P_{k,1/2}$, which is enough for our purpose.

\paragraph{Circuit construction.}
To explain the ingredients of our circuit construction, we first review a classical result of Fredman and Volper~\cite{FV82} on data structures for partial match, which can be implicitly viewed as an efficient $\PAR$ circuit for $P_{n,s}$.

There are two trivial circuits, given by the factorizations $P_{n,s}=P_{n,s}\cdot I$ and $P_{n,s}=I\cdot P_{n,s}$. In the first one, every string enumerates the $\binom{n}{sn}$ possible wildcard masks; in the second one, every query enumerates all $2^{(1-s)n}$ ways to fill its wildcards. The first one gives $l+r\le\h(s)$, and the second one gives $l+r\le1-s$.

In the regime $s\le1/2$, Fredman and Volper combine the two circuits as follows.
\begin{enumerate}[label=(\arabic*)]
  \item Pick a family $\mathscr{B}$ of subsets of $[n]$ of size $2sn$ such that every subset $S$ of size $sn$ is contained in some $B\in\mathscr{B}$. Assign to each $S$ one such $B$.
  \item For each $B\in\mathscr{B}$, construct a circuit for the submatrix formed by the columns of the queries $\bfy$ whose set $S(\bfy)$ of non-wildcard positions is assigned to $B$. This circuit has $2^{2sn}$ middle gates, corresponding to all assignments of values to the positions in $B$. Each string is linked to the unique assignment it agrees with, and each query is linked to the $2^{sn}$ assignments that agree with it on $S(\bfy)$, that is, to all ways of filling the remaining positions $B\setminus S(\bfy)$.
\end{enumerate}
Since every query is handled by exactly one $B$, the right degree is $2^{sn}$. Each $B$ contributes $1$ to the degree of every string, so the left degree is $\lvert\mathscr{B}\rvert$. By a classical probabilistic argument, all sets $S$ can be covered by roughly $\binom{n}{sn}/\binom{2sn}{sn}\approx2^{(\h(s)-2s)n}$ sets $B$.

This gives $l+r\le\h(s)-s$, which improves on both trivial circuits for every $0<s<1/2$. One can view this construction as using a \emph{random covering} argument to combine the advantages of the two trivial factorizations. At $s=1/2$, however, the only choice is $B=[n]$, and the bound degenerates to the trivial $l+r=1/2$.

Interestingly, this construction, like the two trivial ones, is a $\PAR$ circuit, and in particular also an $\OR$ circuit. Although $\OR$ circuits are often the easier ones to construct, for example for the disjointness matrix~\cite{NW21,AL25}, they do not help for the partial match matrix in our regime: a simple combinatorial argument shows that they cannot beat $l+r=1/2$ (Theorem~\ref{thm:or-lower-bound}).

Therefore, to improve the case $s=1/2$ that we are interested in, we need to genuinely exploit algebraic cancellation. Besides the two trivial factorizations, we introduce a third one, coming from the \emph{inclusion--exclusion principle}. For a query $\bfy$, let $A_0(\bfy)$ and $A_1(\bfy)$ be the sets of positions where $\bfy$ is $\mathtt0$ and $\mathtt1$, respectively. Then whether a binary string $\bfx$ matches $\bfy$ is given by the polynomial
\[
 \prod_{i\in A_1(\bfy)}x_i\prod_{i\in A_0(\bfy)}(1-x_i)=\sum_{A_1(\bfy)\subseteq U\subseteq A_1(\bfy)\cup A_0(\bfy)}(-1)^{\lvert U\rvert-\lvert A_1(\bfy)\rvert}\bfx^U,
\]
where $\bfx^U\coloneqq\prod_{i\in U}x_i$. Viewing each monomial $\bfx^U$ as a middle gate, the right-hand side is a new factorization of the partial match matrix, now with coefficients $\pm1$: the string $\bfx$ is linked to all subsets of its support, and the query $\bfy$ to all sets $U$ in the above sum. Combining this factorization with a random covering argument and feeding the result into the rebalancing machinery, we can, perhaps surprisingly, beat $l+r=1/2$.

\subsection{Comparison with the OpenAI preprint}\label{sec:comparison}

In this subsection, we briefly discuss the difference between our approach and the approach of the recent OpenAI preprint~\cite{OpenAI26}. The OpenAI preprint also follows the representation technique framework developed in~\cite{HJ10,MNPW19,RW26}. The technical twist, however, is that their approach \emph{bypasses} the compatibility testing problem altogether.

The high-level approach of the OpenAI preprint can be described as follows. First, pick a large random prime $p\approx2^{(1/2+\varepsilon)n}$. Then use a procedure~(1) to estimate a checksum over all modular solutions, that is, over all $\bfc\in\{0,1\}^n$ with $\bfc\cdot\bfz\equiv\boldsymbol{\tau}\pmod p$. Next, use another procedure~(2) to find all $\bfc$ with $\bfc\cdot\bfz=\boldsymbol{\tau}+kp$ for some $k\in\mathbb{Z}\setminus\{0\}$; these are called \emph{aliases} in the preprint. Finally, subtracting the contribution of the aliases from the checksum leaves the checksum of the integral solutions, and this detects whether an integral solution exists. We also note that, to make the algorithm work, procedure~(1) actually only needs to estimate the checksum rather than compute it exactly; this step also contains technically involved ingredients.

Only procedure~(2) uses the representation technique, and in their algorithm there is no need to design a faster compatibility testing algorithm. Recall that in the usual representation technique, compatibility testing is required to sift out the pseudosolutions when one wants to find a true solution, because there is no good bound on the number of pseudosolutions. In their setting, however, the randomness of $p$ allows them to bound the number of pseudosolutions representing the aliases, so they can afford to enumerate all of these pseudosolutions.

Compared with the OpenAI result, which avoids compatibility testing, our improvement comes from tackling the compatibility testing problem directly. To achieve this, we use the machinery of depth-2 linear circuits, together with new constructions of such circuits.

\subsection{Organization}

\hspace*{\parindent}\mbox{Section~\ref{sec:reduction}\ }reduces subset sum to an algorithm for ``half-stars'' partial match.

\mbox{Section~\ref{sec:circuits}\ }constructs such an algorithm from depth-2 linear circuits. Section~\ref{sec:circuit-prelim} introduces our notation for depth-2 linear circuits. Section~\ref{sec:circuit-enumeration} establishes the precise form of uniformization that we need for the partial match matrices. Section~\ref{sec:colored-rebalancing} proves our degree version of rebalancing. Section~\ref{sec:integer-circuit} gives the circuit construction and bounds its degrees with the help of rebalancing. Finally, Section~\ref{sec:partial-match-circuits} shows that $\OR$ circuits cannot improve on the trivial bound $l+r=1/2$; this justifies our use of algebraic cancellation.

\mbox{Section~\ref{sec:together}\ }combines the two with explicit parameters to prove Theorem~\ref{thm:main}: subset sum can be solved by a randomized algorithm in worst-case time $O(2^{(1/2-1/450000)n})$.

\section{Reduction}\label{sec:reduction}

Formally, an instance of subset sum consists of integers $\bfz=(z_1,\ldots,z_n)\in\mathbb{Z}^n$ and a target $\boldsymbol{\tau}\in\mathbb{Z}$, and the task is to find $\bfc\in\{0,1\}^n$ with $\bfc\cdot\bfz=\boldsymbol{\tau}$, or to report that no such vector exists. We assume that the input length is $2^{o(n)}$ bits, and the running times of our algorithms are measured in bit operations.

In this section, we show that a sufficiently fast algorithm for a special case of partial match yields a subset sum algorithm beating $2^{n/2}$. Section~\ref{sec:pm-prelim} defines the partial match problem we need, and Section~\ref{sec:subset-sum-reduction} gives the reduction. Throughout this section, $\h(t)\coloneqq-t\log_2t-(1-t)\log_2(1-t)$ denotes the binary entropy function, with $0\log_20\coloneqq0$.

\subsection{Partial match}\label{sec:pm-prelim}

The compatibility testing problem arising from our representation is governed by the matrix
\[
P_1\coloneqq\begin{pmatrix}1&0&1\\0&1&1\end{pmatrix},
\]
which we call the \emph{partial match} matrix. Its rows are indexed by $0,1$ and its columns by $0,1,*$.

We write $\otimes$ for the \emph{Kronecker product} of matrices. In general, for matrices $A$ with rows indexed by $I$ and columns by $J$, and $B$ with rows indexed by $I'$ and columns by $J'$, the Kronecker product $A\otimes B$ has rows indexed by $I\times I'$ and columns by $J\times J'$, and is defined by
\[
(A\otimes B)_{(i,i'),(j,j')}\coloneqq A_{i,j}\cdot B_{i',j'}.
\]
Thus the rows of $P_1^{\otimes n}$ are indexed by \emph{binary strings} $\bfx\in\{0,1\}^n$, its columns by \emph{query strings} $\bfy\in\{0,1,*\}^n$, and $(P_1^{\otimes n})_{\bfx,\bfy}=\prod_{i=1}^n(P_1)_{x_i,y_i}$.

This notion lets us describe the matching of strings. For $\bfy\in\{0,1,*\}^n$, put $S(\bfy)\coloneqq\{i:y_i\ne *\}$, the set of \emph{specified} positions of $\bfy$. Then $(P_1^{\otimes n})_{\bfx,\bfy}=1$ if and only if $x_i=y_i$ for every $i\in S(\bfy)$, that is, $\bfx$ matches every specified position of $\bfy$; in this case we say that $\bfx$ \emph{matches} $\bfy$.

The reduction only needs the special case of partial match in which every query has exactly half of its positions specified.

\begin{definition}[Half-stars partial match]\label{def:pm-scheme}
In the \emph{half-stars partial match} problem, we are given an even integer $n$ and two sets
\[
\boldsymbol{A}\subseteq\{0,1\}^n,\qquad\boldsymbol{B}\subseteq\bigl\{\bfy\in\{0,1,*\}^n:\lvert S(\bfy)\rvert=n/2\bigr\}.
\]
The task is to decide whether some $\bfx\in\boldsymbol{A}$ matches some $\bfy\in\boldsymbol{B}$, that is, whether $(P_1^{\otimes n})_{\bfx,\bfy}=1$ for some $(\bfx,\bfy)\in\boldsymbol{A}\times\boldsymbol{B}$.
\end{definition}

The name ``half-stars'' refers to the requirement that exactly half of the positions of every query are wildcards.

We measure an algorithm for this problem by the work it spends on each string of either set. We call a pair $(l,r)\in\mathbb{R}_{\ge0}^2$ an \emph{admissible exponent pair} if for every $\varepsilon>0$ there is a deterministic algorithm solving half-stars partial match in time
\[
O\!\left(2^{(l+\varepsilon)n}\lvert\boldsymbol{A}\rvert+2^{(r+\varepsilon)n}\lvert\boldsymbol{B}\rvert\right).
\]
Clearly admissibility is closed upward: if $(l,r)$ is an admissible exponent pair, $l'\ge l$ and $r'\ge r$, then $(l',r')$ is also admissible.

For example, $(0,1/2)$ is an admissible exponent pair: sort $\boldsymbol{A}$, and for each $\bfy\in\boldsymbol{B}$, enumerate the $2^{n/2}$ binary strings matching $\bfy$ and look each of them up in $\boldsymbol{A}$. The main result of this section, Theorem~\ref{thm:subset-sum-transfer}, shows that improving the sum $l+r$ below $1/2$ by any positive constant already beats $2^{n/2}$ for subset sum.

\subsection{The reduction}\label{sec:subset-sum-reduction}

\begin{theorem}\label{thm:subset-sum-transfer}
Suppose $(l,r)$ is an admissible exponent pair and $l+r\le1/2-2\xi$ for some $\xi>0$. For every $\varepsilon>0$, subset sum has a randomized algorithm in time $O(2^{(1/2-\xi/9+\varepsilon)n})$.
\end{theorem}

We fix some notation. From now on, we assume that $6\mid n$: appending at most five zeros to $\bfz$ achieves this, and changes neither the set of solutions on the original coordinates nor the running time by more than a constant factor. For a subset sum instance $(\bfz,\boldsymbol{\tau})$, the \emph{prefix} consists of the first $n/3$ coordinates, and $N_{\mathrm{sum}}$ denotes the number of distinct subset sums of $z_1,\ldots,z_{n/3}$. For $\bfa,\bfb\in\mathbb{Z}^n$, the \emph{left key} of $\bfa$ is $\bfa\cdot\bfz$ and the \emph{right key} of $\bfb$ is $\boldsymbol{\tau}-\bfb\cdot\bfz$. Thus $\bfa+\bfb$ is a solution if and only if the two keys are equal and $\bfa+\bfb\in\{0,1\}^n$.

We will technically use the following hash family~\cite{Die96}. For a power of two $p$, consider the family
\[
\calH_p\coloneqq\{h_{\alpha,\beta}:\alpha,\beta\in\mathbb{Z}_p\},\qquad h_{\alpha,\beta}(\sfk)\coloneqq(\alpha\sfk+\beta)\bmod p\qquad(\sfk\in\mathbb{Z}).
\]
For a power of two $q\le p$, put $\calI_q\coloneqq\{0,\ldots,q-1\}\subseteq\mathbb{Z}_p$. We filter an integer $\sfk$ by keeping it if and only if $h(\sfk)\in\calI_q$, where $h$ is drawn uniformly from $\calH_p$, that is, $\alpha$ and $\beta$ are independent and uniform in $\mathbb{Z}_p$.

\begin{lemma}\label{lem:filter}
Let $q\le p$ be powers of two, and let $h$ be uniform in $\calH_p$.
\begin{enumerate}
 \item For every $\sfk\in\mathbb{Z}$, we have $\Pr[h(\sfk)\in\calI_q]=q/p$.
 \item For all $\sfk_1,\sfk_2\in\mathbb{Z}$ with $0<\lvert\sfk_1-\sfk_2\rvert<q$, the events $h(\sfk_1)\in\calI_q$ and $h(\sfk_2)\in\calI_q$ are independent.
\end{enumerate}
\end{lemma}

\begin{proof}[Proof of Lemma~\ref{lem:filter}]
For the first claim, $h_{\alpha,\beta}(\sfk)$ is uniform in $\mathbb{Z}_p$ for every fixed $\alpha$, since $\beta$ is uniform.

For the second claim, put $g\coloneqq\gcd(\sfk_2-\sfk_1,p)$. Then $g\le\lvert\sfk_2-\sfk_1\rvert<q$, and $g\mid q$ since both are powers of two. Conditional on $h_{\alpha,\beta}(\sfk_1)=t$, the multiplier $\alpha$ remains uniform, so $h_{\alpha,\beta}(\sfk_2)=t+\alpha(\sfk_2-\sfk_1)$ is uniform on the coset $t+g\mathbb{Z}_p$. Exactly $q/g$ of its $p/g$ elements lie in $\calI_q$. Hence $\Pr[h(\sfk_2)\in\calI_q\mid h(\sfk_1)=t]=q/p$ for every $t$.
\end{proof}

\begin{proof}[Proof of Theorem~\ref{thm:subset-sum-transfer}]
We first enumerate all $2^{n/3}$ subset sums of the prefix by brute force, in time $2^{n/3+o(n)}$. This determines $N_{\mathrm{sum}}$, and we retain one representative prefix for each of the $N_{\mathrm{sum}}$ distinct sums. Depending on the value of $N_{\mathrm{sum}}$, we then run one of the following two procedures.
\begin{itemize}
 \item \emph{Few distinct sums}~\cite[Lemma 3.2]{AKKN16}, when $N_{\mathrm{sum}}\le2^{(1-2\xi/3)n/3}$. Combine the representatives with $\ell\coloneqq\lfloor n/3-\log_2N_{\mathrm{sum}}/2\rfloor$ outside coordinates, and apply meet-in-the-middle \cite{HS74} against the remaining coordinates. The two lists have sizes $N_{\mathrm{sum}}2^\ell$ and $2^{2n/3-\ell}$. Including preprocessing and integer arithmetic, the time is
 \[
 2^{o(n)}2^{n/2}\sqrt{\frac{N_{\mathrm{sum}}}{2^{n/3}}}.
 \]
 \item \emph{Many distinct sums}, when $N_{\mathrm{sum}}>2^{(1-2\xi/3)n/3}$. This is the main case, which occupies the rest of the proof.
\end{itemize}

\paragraph{Representation.}
As in the representation technique~\cite{HJ10,AKKN16,MNPW19,RW26}, we search for $\bfc=\bfa+\bfb$, but we use our representations only on the prefix. For $s\in\{0,3/n,\ldots,1/2\}$ and a parameter $\ell$ determined below, $\bfa$ and $\bfb$ range over the conceptual lists
\[
\calA\coloneqq\{0,1\}^{n/3}\times\{0,1\}^{\ell}\times0^{2n/3-\ell},\qquad
\calB\coloneqq T_{n/3,s}\times0^{\ell}\times\{0,1\}^{2n/3-\ell},
\]
where $T_{m,s}\subseteq\{-1,0,1\}^m$ is the set of vectors with exactly $sm$ nonzero entries. We run the procedure for every such $s$, and we call the two blocks of coordinates after the prefix, of lengths $\ell$ and $2n/3-\ell$, the \emph{outside blocks}. Only the elements of $\calA$ and $\calB$ passing the filter below, which we call \emph{records}, will be actually generated.

Next, we translate compatibility into partial match. In a prefix coordinate, $a_i+b_i\in\{0,1\}$ precisely when $a_i$ matches the query symbol $1,*,0$ for $b_i=-1,0,1$, respectively. On the outside coordinates, $\bfa$ and $\bfb$ have disjoint supports, so $\bfa+\bfb$ is always binary there.

To make the queries half-stars, prepend $(1-2s)n/3$ dummy pairs $(1,-1)$, which always match. For $\bfa\in\calA$ and $\bfb\in\calB$, write $\bfx(\bfa)$ and $\bfy(\bfb)$ for the resulting binary and query strings. They have even length $n'\coloneqq(2-2s)n/3$, exactly $(1-s)n/3=n'/2$ positions of $\bfy(\bfb)$ are specified, and
\[
\bfa+\bfb\in\{0,1\}^n\iff\bfx(\bfa)\text{ matches }\bfy(\bfb).
\]
For strings of length $n'$, the partial match algorithm runs in time $O(\operatorname{cost}(\lvert\boldsymbol{A}\rvert,\lvert\boldsymbol{B}\rvert))$, where
\[
\operatorname{cost}(N_{\sfL},N_{\sfR})\coloneqq2^{(l+\varepsilon)n'}N_{\sfL}+2^{(r+\varepsilon)n'}N_{\sfR}.
\]

We choose $\ell\in[0,n/2]$ as the largest even integer that is not greater than
\[
\frac{n}{6}\left(1+\h(s)+s-2(l-r)(1-s)\right).
\]
This is possible since $\lvert l-r\rvert\le1/2$, so the expression itself lies in $[0,n/2]$. Since $6\mid n$, both of the outside blocks must have even length. The following lemma bounds the cost of running partial match directly on all of $\calA$ and $\calB$, which we call the \emph{unscreened} cost.

\begin{lemma}\label{lem:unscreened}
We have
\[
\operatorname{cost}\bigl(\lvert\calA\rvert,\lvert\calB\rvert\bigr)\le2^{n/2+(1-\xi)n/3+2\varepsilon n/3+O(1)}.
\]
\end{lemma}

\begin{proof}[Proof of Lemma~\ref{lem:unscreened}]
Using $\binom{n/3}{sn/3}\le2^{\h(s)n/3}$, the list sizes are
\[
\lvert\calA\rvert=2^{n/3+\ell},\qquad \lvert\calB\rvert=\binom{n/3}{sn/3}2^{sn/3+2n/3-\ell}\le2^{(\h(s)+s)n/3+2n/3-\ell}.
\]
With $n'=2(1-s)n/3$, the two terms of the cost thus have exponents at most
\[
2(1-s)(l+\varepsilon)\frac{n}{3}+\frac{n}{3}+\ell
\qquad\text{and}\qquad
2(1-s)(r+\varepsilon)\frac{n}{3}+(\h(s)+s)\frac{n}{3}+\frac{2n}{3}-\ell.
\]
Our choice of $\ell$ makes them equal up to an additive constant, so both are at most their average plus $O(1)$. Using $\h(s)\le1$, $l+r\le1/2-2\xi$, and $s\le1/2$, the average is
\begin{align*}
\frac{n}{2}+\left((l+r+2\varepsilon)(1-s)+\frac{\h(s)+s}{2}\right)\frac{n}{3}
&\le\frac{n}{2}+\left(\frac{\h(s)+1}{2}-2\xi(1-s)\right)\frac{n}{3}+\frac{2\varepsilon n}{3}\\
&\le\frac{n}{2}+\frac{(1-\xi)n}{3}+\frac{2\varepsilon n}{3}.\qedhere
\end{align*}
\end{proof}

Now fix a solution $\bfc$. For a given $s$, we call an integer $\sfk$ a \emph{correct key} if it is the left key of some $\bfa\in\calA$ with $\bfc-\bfa\in\calB$. In this case, $\sfk$ is also the right key of $\bfc-\bfa$, since $\bfc\cdot\bfz=\boldsymbol{\tau}$.

\begin{lemma}\label{lem:good-layer}
For some $s\in\{0,3/n,\ldots,1/2\}$, there are at least $N_{\mathrm{sum}}/(n/3+2)$ distinct correct keys.
\end{lemma}

\begin{proof}[Proof of Lemma~\ref{lem:good-layer}]
Each binary prefix $\bfa$ determines $\bfb\coloneqq\bfc-\bfa$ on the prefix. Complementing $\bfa$ changes the support size of this $\bfb$ from $k$ to $n/3-k$, and maps its prefix sum to the total prefix sum minus that sum. Consequently, at least $N_{\mathrm{sum}}/2$ distinct prefix sums have a representation with $k\le n/6$: every sum without one has a distinct complementary sum with one. Placing the outside part of $\bfc$ on the designated sides merely translates these sums by a fixed amount. Among the $n/6+1$ values of $s$, one therefore receives at least $N_{\mathrm{sum}}/(n/3+2)$ distinct correct keys.
\end{proof}

We call such a value of $s$ the \emph{good layer}.

\paragraph{Filtering.}
We filter the elements of $\calA$ and $\calB$ by their keys $\sfk$, namely the left key for $\bfa$ and the right key for $\bfb$. Let $p$ be the smallest power of two with $p>2N_{\mathrm{sum}}\bigl(\lvert\boldsymbol{\tau}\rvert+\sum_i\lvert z_i\rvert+1\bigr)$, and put
\[
q\coloneqq\frac{p}{2^{\lfloor\log_2N_{\mathrm{sum}}\rfloor}}.
\]
Draw $h=h_{\alpha,\beta}$ uniformly from $\calH_p$, and keep an element if and only if $h(\sfk)\in\calI_q$ for its key $\sfk$; we then say that the element and its key \emph{survive}.

\begin{lemma}\label{lem:filter-success}
Every element of $\calA$ and $\calB$ survives with probability less than $2/N_{\mathrm{sum}}$. In the good layer, some correct key survives with probability at least $1/(n/3+3)$.
\end{lemma}

\begin{proof}[Proof of Lemma~\ref{lem:filter-success}]
By Lemma~\ref{lem:filter}, each key survives with probability $q/p=2^{-\lfloor\log_2N_{\mathrm{sum}}\rfloor}\in[1/N_{\mathrm{sum}},2/N_{\mathrm{sum}})$, which gives the first claim. Since $q\ge p/N_{\mathrm{sum}}>2\bigl(\lvert\boldsymbol{\tau}\rvert+\sum_i\lvert z_i\rvert+1\bigr)$, any two keys differ by at most $\lvert\boldsymbol{\tau}\rvert+2\sum_i\lvert z_i\rvert<q$. By Lemma~\ref{lem:filter}, the survival events of distinct keys are therefore pairwise independent.

Let $X$ count the surviving distinct correct keys in the good layer. By Lemma~\ref{lem:good-layer}, $\mathbb{E}[X]\ge1/(n/3+2)$, and by pairwise independence, $\mathbb{E}[X^2]\le\mathbb{E}[X]+(\mathbb{E}[X])^2$. Cauchy--Schwarz gives
\[
\Pr[X>0]\ge\frac{(\mathbb{E}[X])^2}{\mathbb{E}[X^2]}\ge\frac{1}{1+1/\mathbb{E}[X]}\ge\frac{1}{n/3+3}.\qedhere
\]
\end{proof}

Next, we generate the surviving records without listing all of $\calA$ and $\calB$, using the meet-in-the-middle enumeration idea of \cite[Lemma 3.8]{MNPW19}.

\begin{lemma}\label{lem:generate}
The surviving records can be listed in time $2^{o(n)}\bigl(2^{5n/12}+2^{(1/2-\xi/6)n}+N_{\sfL}+N_{\sfR}\bigr)$, where $N_{\sfL}$ and $N_{\sfR}$ are the numbers of surviving records in $\calA$ and $\calB$, respectively.
\end{lemma}

\begin{proof}[Proof of Lemma~\ref{lem:generate}]
Split each of the three blocks into equal halves. Then every $\bfa\in\calA$ is uniquely a sum $\bfa_1+\bfa_2$ of partial assignments on the unions of the first and second halves of the three blocks, respectively, and we enumerate these into the \emph{left half-lists} $\calA_1$ and $\calA_2$. Likewise $\bfb=\bfb_1+\bfb_2$, giving the \emph{right half-lists} $\calB_1$ and $\calB_2$; here we keep only half-prefixes with at most $sn/3$ nonzero entries, and group the elements of $\calB_1$ by the number of nonzero entries in their half-prefixes.

Keys are additive over the halves. For $\bfa=\bfa_1+\bfa_2$, put $\sfk_i\coloneqq\bfa_i\cdot\bfz$, so that $\bfa$ survives if and only if $\alpha\sfk_1\bmod p$ lies in the cyclic interval $\calI_q-\beta-\alpha\sfk_2$. So we sort $\calA_1$ by $\alpha\sfk_1\bmod p$, and for each $\bfa_2\in\calA_2$ report the entries of $\calA_1$ in this interval. On the right, with $\sfk_i\coloneqq\bfb_i\cdot\bfz$, the interval is $\alpha\boldsymbol{\tau}+\beta-\alpha\sfk_2-\calI_q$; for each $\bfb_2\in\calB_2$, we search only the group of $\calB_1$ whose members give exactly $sn/3$ nonzero entries in the prefix when added to $\bfb_2$. Each interval is a union of at most two ordinary intervals. Preserve all assignments, including equal residues.

This takes time proportional to the sizes of the four half-lists plus the reported records, times $2^{o(n)}$ for arithmetic and sorting. The left half-lists $\calA_1,\calA_2$ have at most $2^{5n/12+O(1)}$ entries each, since $\ell\le n/2$. For either right half-list, the number of possible half-prefixes is at most
\[
\sum_{j=0}^{sn/3}\binom{n/6}{j}2^j\le2^{n/6+sn/3}.
\]
Each of $\calB_1,\calB_2$ also has half of the $2n/3-\ell$ binary variables in the second outside block, so substituting $\ell$ gives the size exponent
\[
\begin{aligned}
\frac{n}{6}+\frac{sn}{3}+\frac{1}{2}\!\left(\frac{2n}{3}-\ell\right)
&=\frac{5n}{12}-\frac{\h(s)n}{12}+\frac{sn}{4}+\frac{(l-r)(1-s)n}{6}+O(1)\\
&\le\frac{n}{2}-\frac{\h(s)n}{12}+\frac{sn}{6}-\frac{\xi(1-s)n}{3}+O(1)\\
&\le\left(\frac{1}{2}-\frac{\xi}{6}\right)\!n+O(1),
\end{aligned}
\]
using $l-r\le1/2-2\xi$ in the second step, and $\h(s)\ge2s$ and $s\le1/2$ in the last.
\end{proof}

Finally, we test the surviving records by half-stars partial match.

\begin{lemma}\label{lem:detect}
Given the surviving records, we can find a solution among the pairs of surviving records, or correctly report that none exists, in time $2^{o(n)}\operatorname{cost}(N_{\sfL},N_{\sfR})$.
\end{lemma}

\begin{proof}[Proof of Lemma~\ref{lem:detect}]
Sort the surviving records by their keys; since keys have $2^{o(n)}$ bits, this costs a factor of $2^{o(n)}$ per record. For each key value $\sfk$ occurring on both sides, let
\[
\begin{aligned}
\boldsymbol{A}_{\sfk}&\coloneqq\{\bfx(\bfa):\bfa\text{ survives and has left key }\sfk\},\\
\boldsymbol{B}_{\sfk}&\coloneqq\{\bfy(\bfb):\bfb\text{ survives and has right key }\sfk\},
\end{aligned}
\]
and run the half-stars partial match algorithm on $(\boldsymbol{A}_{\sfk},\boldsymbol{B}_{\sfk})$. By the translation to partial match, some call answers positively if and only if there are surviving $\bfa,\bfb$ with
\[
\bfa\cdot\bfz=\boldsymbol{\tau}-\bfb\cdot\bfz\qquad\text{and}\qquad\bfa+\bfb\in\{0,1\}^n,
\]
that is, if and only if some pair of surviving records forms a solution. The sets $\boldsymbol{A}_{\sfk}$ have total size at most $N_{\sfL}$, and the sets $\boldsymbol{B}_{\sfk}$ total size at most $N_{\sfR}$, so the total cost of these calls is $O(\operatorname{cost}(N_{\sfL},N_{\sfR}))$.

If the call for $\sfk$ answers positively, repeatedly bisect $\boldsymbol{A}_{\sfk}$ and $\boldsymbol{B}_{\sfk}$, always retaining a half for which the answer remains positive. After $O(n)$ such calls, one matching pair $(\bfx(\bfa),\bfy(\bfb))$ remains, and $\bfa+\bfb$ is a solution.
\end{proof}

\paragraph{Putting things together.}
An \emph{attempt} in layer $s$ draws $h$, generates the surviving records by Lemma~\ref{lem:generate}, and runs Lemma~\ref{lem:detect}. Since $\operatorname{cost}(N_{\sfL},N_{\sfR})\ge N_{\sfL}+N_{\sfR}$, the last term of Lemma~\ref{lem:generate} is covered by Lemma~\ref{lem:detect}. By Lemma~\ref{lem:filter-success} and linearity of expectation, $\mathbb{E}[N_{\sfL}]<2\lvert\calA\rvert/N_{\mathrm{sum}}$ and $\mathbb{E}[N_{\sfR}]<2\lvert\calB\rvert/N_{\mathrm{sum}}$, so filtering reduces the expected cost by a factor of at least $N_{\mathrm{sum}}/2$ compared to the unscreened cost. By Lemma~\ref{lem:unscreened}, the expected time of one attempt is therefore at most
\[
2^{o(n)}\left(2^{5n/12}+2^{(1/2-\xi/6)n}+\frac{2^{n/2+(1-\xi)n/3+2\varepsilon n/3}}{N_{\mathrm{sum}}}\right).
\]
This bound is known in advance, and we abort an attempt after $n^2$ times this bound; by Markov's inequality, this loses at most $1/n^2$ success probability.

If a solution exists, then by Lemmas~\ref{lem:filter-success} and~\ref{lem:detect}, an attempt in the good layer finds a solution with probability at least $1/(n/3+3)$ if not aborted, hence before the cutoff with probability at least $1/(n/3+3)-1/n^2\ge1/(2n)$. We perform $n^2$ independent attempts per layer, so the failure probability is exponentially small.

It remains to compare the two procedures. At the threshold $N_{\mathrm{sum}}=2^{(1-2\xi/3)n/3}$, both principal costs are $2^{(1/2-\xi/9)n}$, because
\[
\sqrt{\frac{N_{\mathrm{sum}}}{2^{n/3}}}=\frac{2^{(1-\xi)n/3}}{N_{\mathrm{sum}}}=2^{-\xi n/9}.
\]
The exponent of the right half-lists is smaller by $\xi/18$, and preprocessing and the left half-lists are also cheaper. Both procedures therefore run in time $2^{(1/2-\xi/9+2\varepsilon/3)n+o(n)}$.
\end{proof}

\section{Depth-2 linear circuits}\label{sec:circuits}

In this section, we obtain a nontrivial partial match algorithm by designing depth-2 linear circuits.

\subsection{Preliminaries}\label{sec:circuit-prelim}

We first fix our notation for depth-2 linear circuits, following~\cite[Section 1.2]{JS13} and~\cite[Section 3.2]{AL25}. A \emph{depth-2 linear circuit} is a pair of matrices $C\coloneqq(\bfL_C,\bfR_C)$ whose product $\bfL_C\bfR_C$ is defined, and it computes the map $\bfu\mapsto\bfu\bfL_C\bfR_C$. The rows of $\bfL_C$ index the \emph{input gates}, the columns of $\bfL_C$, equivalently the rows of $\bfR_C$, index the \emph{middle gates}, and the columns of $\bfR_C$ index the \emph{output gates}. We write $\Mid(C)$ for the set of middle gates. The nonzero entries of $\bfL_C$ and $\bfR_C$ are the \emph{edges} of $C$, and their values are the \emph{edge coefficients}.

We measure a circuit by its degrees. Write $\nnz$ for the number of nonzero entries. For an input gate $\bfx$ and an output gate $\bfy$, put
\[
\deg_{\sfL}(C,\bfx)\coloneqq\nnz(\bfL_C[\bfx,*]),\qquad\deg_{\sfR}(C,\bfy)\coloneqq\nnz(\bfR_C[*,\bfy]).
\]
Thus $\deg_{\sfL}(C,\bfx)$ and $\deg_{\sfR}(C,\bfy)$ count the edges at the gate. We write $\deg_{\sfL}(C)$ and $\deg_{\sfR}(C)$ for their maxima over all gates, and call $\nnz(\bfL_C)+\nnz(\bfR_C)$ the \emph{size} of $C$.

Different requirements on the product $\bfL_C\bfR_C$ give different circuit models. All products are taken over the integers, and for a $01$ matrix $M$ we consider the following four models.

\begin{center}
\renewcommand{\arraystretch}{1.15}
\begin{tabular}{@{}p{0.13\linewidth}p{0.21\linewidth}p{0.58\linewidth}@{}}
\toprule
Model & Allowed coefficients & Requirement for computing a $01$ matrix $M$ \\
\midrule
$\PAR$ & $0,1$ & $\bfL_C\bfR_C=M$ \\
$\OR$ & $0,1$ & $(\bfL_C\bfR_C)_{\bfx,\bfy}\ne0\iff M_{\bfx,\bfy}=1$ \\
$\EXACTZ$ & all integers & $\bfL_C\bfR_C=M$ \\
$\SUPZ$ & all integers & $(\bfL_C\bfR_C)_{\bfx,\bfy}\ne0\iff M_{\bfx,\bfy}=1$ \\
\bottomrule
\end{tabular}
\end{center}

Here $\PAR$ stands for a partition of the ones of $M$ into all-one rectangles. Our constructions are in the third model: $C$ is an \emph{exact integer circuit}, or $\EXACTZ$ circuit, for $M$ if both matrices have integer entries and $\bfL_C\bfR_C=M$.

Next, we specialize to partial match. Recall from Section~\ref{sec:pm-prelim} that the rows of $P_1^{\otimes n}$ are indexed by binary strings and its columns by query strings. Accordingly, we name the input gates of a circuit for $P_1^{\otimes n}$ by binary strings $\bfx$ and its output gates by query strings $\bfy$; the string coordinates $x_i$ should not be confused with the input values $u_{\bfx}$. The entries of $P_1^{\otimes n}$ are then polynomials in the coordinates of $\bfx$: for $\bfy\in\{0,1,*\}^n$, put
\[
f_{\bfy}(\bfx)\coloneqq\prod_{i:y_i=1}x_i\prod_{i:y_i=0}(1-x_i),
\]
so that $(P_1^{\otimes n})_{\bfx,\bfy}=f_{\bfy}(\bfx)$. Since half-stars partial match only involves queries with a fixed number of specified positions, we let $P_{n,s}$ be the submatrix of $P_1^{\otimes n}$ formed by the columns $\bfy$ with $\lvert S(\bfy)\rvert=\lfloor sn\rfloor$, retaining all rows.

We measure circuits for $P_{n,s}$ by the growth rate of their degrees, in the same way as admissible exponent pairs. Let $\calK(s)$ consist of the pairs $(l,r)\in\mathbb{R}_{\ge0}^2$ for which there are $\EXACTZ$ circuits $C_n$ computing $P_{n,s}$ with $\deg_{\sfL}(C_n)\le2^{(l+o(1))n}$ and $\deg_{\sfR}(C_n)\le2^{(r+o(1))n}$. Equivalently, for every $\varepsilon>0$, every sufficiently large $n$ has an $\EXACTZ$ circuit for $P_{n,s}$ with degrees at most $2^{(l+\varepsilon)n}$ and $2^{(r+\varepsilon)n}$. For $s=1/2$, even lengths suffice: for odd $n$, appending $0$ to both strings embeds $P_{n,1/2}$ into $P_{n+1,1/2}$ without increasing either degree.

Finally, we fix some notation for counting. For $0\le v\le u$, put
\[
\h(u,v)\coloneqq u\log_2u-v\log_2v-(u-v)\log_2(u-v),
\]
with $0\log_20\coloneqq0$; thus $\h(t)=\h(1,t)$ is the binary entropy function of Section~\ref{sec:reduction}. For a string $\bfx$ and a symbol $i$, let $\#i(\bfx)$ count the occurrences of $i$ in $\bfx$.

\subsection{From circuits to algorithms}\label{sec:circuit-enumeration}

For an algorithm, the mere existence of good circuits is not enough: we must be able to \emph{explicitly construct} them, in the sense that the edges at any given input or output gate can be listed in time comparable to its degree. Lemma~\ref{lem:block-grouping} reduces circuits for $P_{n,1/2}$ to tensor powers of a circuit for $P_{k,1/2}$. When $k$ is fixed, the latter circuit can be hardcoded, and the construction becomes explicit for all $n$ (Corollary~\ref{cor:enumeration}) and hence partial match algorithms (Corollary~\ref{cor:circuit-to-pm}).

We start with the combinatorial part, which does not yet care about explicitness. A circuit for $P_{k,1/2}$ only handles queries whose wildcards are spread evenly over blocks of length $k$, while a general query of $P_{n,1/2}$ may have its wildcards distributed arbitrarily. The following lemma shows that we can always permute the coordinates so that the wildcards become evenly spread, and that the price is paid only on the left degree.

\begin{lemma}\label{lem:block-grouping}
Let $k,n\ge2$ be even, and put $d\coloneqq\lceil n/k\rceil$. If $E$ is an $\EXACTZ$ circuit computing $(P_{k,1/2})^{\otimes d}$, then there is an $\EXACTZ$ circuit $C_n$ computing $P_{n,1/2}$ with
\[
\deg_{\sfL}(C_n)\le\binom{kd}{d}\deg_{\sfL}(E),\qquad\deg_{\sfR}(C_n)\le\deg_{\sfR}(E).
\]
\end{lemma}

\begin{proof}[Proof of Lemma~\ref{lem:block-grouping}]
Let $\bfy$ be a length-$kd$ query with exactly $kd/2$ wildcards. We call a permutation of $[kd]$ \emph{balancing} for $\bfy$ if, after dividing the permuted query into $d$ consecutive blocks of length $k$, each block contains exactly $k/2$ wildcards.

Suppose that $\boldsymbol{\iota}$ is a family of permutations containing a balancing permutation for every such query. Given a length-$n$ binary string $\bfx$, append $kd-n$ zeros; given a length-$n$ query $\bfy$, append $(kd-n)/2$ zeros followed by $(kd-n)/2$ wildcards. The new positions all match, and the padded query has exactly $kd/2$ wildcards.

Construct $C_n$ with middle gates $\{(e,\iota)\}_{e\in\Mid(E),\iota\in\boldsymbol{\iota}}$. Each input gate $\bfx$ uses every permutation $\iota$ of its padded binary string: copy the edges from that permuted string to $e$ in $E$ as edges $\bfx\to(e,\iota)$ in $C_n$, with the same coefficients. Each output gate $\bfy$ selects exactly one permutation $\iota(\bfy)\in\boldsymbol{\iota}$ balancing for its padded query: copy the edges from $e$ to that permuted query in $E$ as edges $(e,\iota(\bfy))\to\bfy$ in $C_n$, with the same coefficients. Padding and simultaneous coordinate permutations preserve matching, and only the selected copy contributes to the entry for $(\bfx,\bfy)$. Thus $C_n$ computes $P_{n,1/2}$, with $\deg_{\sfL}(C_n)\le\lvert\boldsymbol{\iota}\rvert\deg_{\sfL}(E)$ and $\deg_{\sfR}(C_n)\le\deg_{\sfR}(E)$.

It remains to construct $\boldsymbol{\iota}$ with $\lvert\boldsymbol{\iota}\rvert\le\binom{kd}{d}$, and to select a balancing permutation for each padded query. Given a length-$kd$ query $\bfy$ with $\lvert S(\bfy)\rvert=kd/2$, start with an empty stack and scan its coordinates $i=1,\ldots,kd$, pushing $i$ onto the stack. Whenever the top $k$ coordinates contain exactly $k/2$ wildcards in $\bfy$, remove them and map them, in increasing coordinate order, to the next block.

There are exactly $d$ removals. At the end of each step, no consecutive $k$ coordinates on the stack contain exactly $k/2$ wildcards: pushing can create only the top window, and removing it leaves a prefix of the previous stack. If the stack were nonempty at the end, its length would be divisible by $k$ and exactly half its coordinates would contain wildcards. Divide it into blocks of length $k$; their wildcard counts have average $k/2$. A sliding window's wildcard count changes by at most $1$ at each step, so some window must contain exactly $k/2$ wildcards, a contradiction. The stack is therefore empty at the end, defining a permutation $\iota(\bfy)$.

The permutation $\iota(\bfy)$ is determined by the $d$ positions in $[kd]$ at which blocks are removed: simulating the stack procedure with these removal positions, ignoring the query characters, recovers it. Let $\boldsymbol{\iota}$ consist of the permutations obtained in this way from all $d$-subsets of $[kd]$ for which every removal can be performed. This family contains every $\iota(\bfy)$ and has $\lvert\boldsymbol{\iota}\rvert\le\binom{kd}{d}$.
\end{proof}

The factor $\binom{kd}{d}\le2^{\h(1/k)kd}$ has exponential rate $\h(1/k)$, which tends to $0$ as $k$ grows. So for any slack we can fix one good circuit $C_0$ for $P_{k,1/2}$ once and for all, and apply Lemma~\ref{lem:block-grouping} to its tensor powers. Since the edges of a tensor power can be listed block by block, the resulting circuits are explicit.

\begin{corollary}\label{cor:enumeration}
Let $(l,r)\in\calK(1/2)$ and $\varepsilon>0$. There are $\EXACTZ$ circuits $C_n$ computing $P_{n,1/2}$, one for each even $n\ge2$, and deterministic algorithms $A_{\sfL}$ and $A_{\sfR}$ with the following properties.
\begin{itemize}
 \item Given $\bfx\in\{0,1\}^n$, the algorithm $A_{\sfL}$ lists all edges $\bfx\to e$ of $C_n$, each as a middle-gate name $e$ with its coefficient, in time $O\bigl(2^{(l+\varepsilon)n}\bigr)$.
 \item Given $\bfy\in\{0,1,*\}^n$ with $\lvert S(\bfy)\rvert=n/2$, the algorithm $A_{\sfR}$ lists all edges $e\to\bfy$ of $C_n$, each as a middle-gate name $e$ with its coefficient, in time $O\bigl(2^{(r+\varepsilon)n}\bigr)$.
\end{itemize}
\end{corollary}

\begin{proof}[Proof of Corollary~\ref{cor:enumeration}]
Since $(l,r)\in\calK(1/2)$, its definition lets us choose and fix an even $k\ge2$ with $\h(1/k)\le\varepsilon/4$ and an $\EXACTZ$ circuit $C_{0}$ computing $P_{k,1/2}$ with $\deg_{\sfL}(C_{0})\le2^{(l+\varepsilon/4)k}$ and $\deg_{\sfR}(C_{0})\le2^{(r+\varepsilon/4)k}$. Hardcode this fixed circuit's adjacency lists and integer edge coefficients into the algorithms.

For each even $n\ge2$, let $d\coloneqq\lceil n/k\rceil$ and let $C_n$ be the circuit obtained by applying Lemma~\ref{lem:block-grouping} to $E\coloneqq C_0^{\otimes d}$. We first list the edges of $C_0^{\otimes d}$ at a binary string or query of length $kd$: divide it into $d$ blocks of length $k$, list the edges of each block in the hardcoded circuit $C_0$, and take the Cartesian product. This gives the names $e=(e_1,\ldots,e_d)$, with edge coefficient equal to the product of the $d$ block coefficients.

Next, recall from the proof of Lemma~\ref{lem:block-grouping} that the middle gates of $C_n$ are the pairs $(e,\iota)$, and that each $\iota\in\boldsymbol{\iota}$ is determined by its $d$ removal positions; we name $\iota$ by the $kd$-bit indicator string of these positions, from which $\iota$ is recovered in time $n^{O(1)}$. Given $\bfx$, the algorithm $A_{\sfL}$ pads it, enumerates the at most $\binom{kd}{d}$ indicator strings, recovers each $\iota$, and lists the edges of $C_0^{\otimes d}$ at the permuted padded string. Given $\bfy$, the algorithm $A_{\sfR}$ pads it, runs the stack procedure once to find $\iota(\bfy)$, and lists the edges of $C_0^{\otimes d}$ at the permuted padded query. Names and edge coefficients have $O(n)$ bits, since each coefficient is a product of $d$ fixed integers. Therefore
\[
\begin{aligned}
\operatorname{time}(A_{\sfL})
&\le n^{O(1)}\binom{kd}{d}\deg_{\sfL}(C_{0})^{d}
\le n^{O(1)}2^{\h(1/k)kd}2^{(l+\varepsilon/4)kd}\\
&\le n^{O(1)}2^{(l+\varepsilon/2)kd}
=O\!\left(2^{(l+\varepsilon)n}\right),\\
\operatorname{time}(A_{\sfR})
&\le n^{O(1)}\deg_{\sfR}(C_{0})^{d}
\le n^{O(1)}2^{(r+\varepsilon/4)kd}
=O\!\left(2^{(r+\varepsilon)n}\right).
\end{aligned}
\]
Here $k$ is fixed, so $kd=n+O(1)$, and the remaining exponent slack absorbs the polynomial factors and padding loss. This proves the corollary.
\end{proof}

Once the edges at every gate can be listed efficiently, we can turn these circuits into a partial match algorithm via ``meet-in-the-middle''.

\begin{corollary}\label{cor:circuit-to-pm}
Every $(l,r)\in\calK(1/2)$ is an admissible exponent pair.
\end{corollary}

\begin{proof}[Proof of Corollary~\ref{cor:circuit-to-pm}]
Use the algorithms $A_{\sfL},A_{\sfR}$ from Corollary~\ref{cor:enumeration}. For each $\bfx\in\boldsymbol{A}$, use $A_{\sfL}$ to enumerate all edges $\bfx\to e$ in $C_{n}$ and add each edge coefficient to the left counter for $e$. For each $\bfy\in\boldsymbol{B}$, use $A_{\sfR}$ to enumerate all edges $e\to\bfy$ in $C_{n}$ and add each edge coefficient to the right counter for $e$. For each $e$, multiply its left and right counters, taking absent counters to be zero, and sum the products. The result is
\[
\sum_{e}
\Bigl(\sum_{\bfx\in\boldsymbol{A}}(\bfL_{C_{n}})_{\bfx,e}\Bigr)
\Bigl(\sum_{\bfy\in\boldsymbol{B}}(\bfR_{C_{n}})_{e,\bfy}\Bigr)
=\sum_{\bfx\in\boldsymbol{A},\,\bfy\in\boldsymbol{B}}(P_{1}^{\otimes n})_{\bfx,\bfy}.
\]
Since $C_{n}$ computes $P_{n,1/2}$ over the integers, this is the number of matching pairs in $\boldsymbol{A}\times\boldsymbol{B}$. Output that a matching pair exists if and only if the result is positive.

Implement the counters by sorting the generated coefficients deterministically by $e$. Middle-gate names and edge coefficients have $O(n)$ bits, and $\lvert\boldsymbol{A}\rvert\le2^n$, $\lvert\boldsymbol{B}\rvert\le3^n$, so sorting and integer arithmetic introduce only a factor of $n^{O(1)}$. Using Corollary~\ref{cor:enumeration} with slack $\varepsilon/2$ gives the required running time.
\end{proof}

\subsection{Colored volume and rebalancing}\label{sec:colored-rebalancing}

In this subsection we prove our degree version of rebalancing. This idea comes from previous works~\cite{JS13,AGP23,Ser22,AL25}, where it was developed for optimizing circuit sizes; here we adapt the method to degrees.

A single circuit constructed in Section~\ref{sec:integer-circuit} does not improve the degrees on its own: its middle gates fall into groups with rather different left and right degrees. The idea of rebalancing is to take tensor powers, and to choose the circuit for the remaining tensor factors separately for each group, so that the two degrees even out.

To make this precise, we track the middle gates in groups given by a coloring. A coloring $\chi\colon\Mid(C)\to\calZ$ assigns to every middle gate a color in the color set $\calZ$; a circuit together with a coloring is a \emph{colored circuit}. We usually suppress the coloring from the notation and simply speak of a colored circuit $C$. For $\zeta\in\calZ$, let $C_\zeta$ retain all input and output gates and just the middle gates of color $\zeta$. That is, $C_\zeta=(\bfL_\zeta,\bfR_\zeta)$, where $\bfL_\zeta$ consists of the columns of $\bfL_C$ and $\bfR_\zeta$ of the rows of $\bfR_C$ indexed by middle gates of color $\zeta$, so that
\[
\bfL_C\bfR_C=\sum_{\zeta\in\calZ}\bfL_\zeta\bfR_\zeta.
\]

\subsubsection{Rebalancing}

Now we sketch the basic idea of rebalancing. Let $C^{(1)},\ldots,C^{(m)}$ be $m$ different colored circuits, and let $C^{(i)}_\zeta=(\bfL^{(i)}_\zeta,\bfR^{(i)}_\zeta)$ denote the $\zeta$-colored part of $C^{(i)}$. Suppose that all of them compute the same desired matrix $M$, that is, the identity
\[
M=\sum_{\zeta\in\calZ}\bfL^{(i)}_\zeta\bfR^{(i)}_\zeta
\]
holds for every $i$. For each single $i$, we can bound the degrees of $C^{(i)}$ by the modest estimate summing over the colors:
\begin{equation}\label{eqn:modest}
\deg_{\sfL}(C^{(i)})\le\sum_{\zeta\in\calZ}\deg_{\sfL}(C^{(i)}_\zeta),\qquad\deg_{\sfR}(C^{(i)})\le\sum_{\zeta\in\calZ}\deg_{\sfR}(C^{(i)}_\zeta).
\end{equation}
Although this estimate is rarely tight, we will see later that it is more convenient for the rebalancing process.

Now consider the question of constructing circuits for the tensor power $M^{\otimes d}$ from the base circuits for $M$ that we know. At first glance, the circuits we can build from the $C^{(i)}$ are the tensor products $C^{(i_1)}\otimes\cdots\otimes C^{(i_d)}$. The main insight of rebalancing is that we actually have a much richer family. Let us first consider the tensor square $M^{\otimes2}$. We may first expand one tensor factor by some chosen circuit $C^{(i)}$, that is,
\[
M^{\otimes2}=\Bigl(\sum_{\zeta\in\calZ}\bfL^{(i)}_\zeta\bfR^{(i)}_\zeta\Bigr)\otimes M=\sum_{\zeta\in\calZ}\Bigl(\bfL^{(i)}_\zeta\bfR^{(i)}_\zeta\Bigr)\otimes M.
\]
After pulling $M$ inside the summation, notice that we actually have the freedom to choose a different circuit computing the inner $M$ for each color $\zeta$. Let $j\colon\calZ\to[m]$ be such a choice; we can then write
\[
M^{\otimes2}=\sum_{\zeta\in\calZ}\Bigl(\bfL^{(i)}_\zeta\bfR^{(i)}_\zeta\Bigr)\otimes\Bigl(\bfL_{C^{(j_\zeta)}}\bfR_{C^{(j_\zeta)}}\Bigr)=\sum_{\zeta\in\calZ}\Bigl(\bfL^{(i)}_\zeta\otimes\bfL_{C^{(j_\zeta)}}\Bigr)\Bigl(\bfR^{(i)}_\zeta\otimes\bfR_{C^{(j_\zeta)}}\Bigr).
\]
By the modest estimate~\eqref{eqn:modest}, this circuit has
\begin{equation}\label{eqn:modest2}
\deg_{\sfL}\le\sum_{\zeta\in\calZ}\deg_{\sfL}(C^{(i)}_\zeta)\deg_{\sfL}(C^{(j_\zeta)}),\qquad\deg_{\sfR}\le\sum_{\zeta\in\calZ}\deg_{\sfR}(C^{(i)}_\zeta)\deg_{\sfR}(C^{(j_\zeta)}).
\end{equation}

In general, we can construct a family $\calF_d$ of depth-2 linear circuits computing $M^{\otimes d}$ in the following recursive way.

First, $\calF_0$ is a singleton, consisting of the trivial depth-2 linear circuit computing the $1$-dimensional identity map, which has degree $1$ on both sides.

Then for $d\ge1$, we construct $\calF_d$ from $\calF_{d-1}$ as follows. Each circuit in $\calF_d$ is given by first selecting an index $i\in[m]$ and a choice function $E\colon\calZ\to\calF_{d-1}$ that chooses a circuit $E_\zeta$ for the $(d-1)$-th power for every color $\zeta$. The circuit is then given by the parts
\begin{equation}\label{eqn:family}
M^{\otimes d}=\sum_{\zeta\in\calZ}\Bigl(\bfL^{(i)}_\zeta\otimes\bfL_{E_\zeta}\Bigr)\Bigl(\bfR^{(i)}_\zeta\otimes\bfR_{E_\zeta}\Bigr).
\end{equation}
We call the circuits in $\calF_d$ the \emph{rebalancings} of the base circuits $C^{(1)},\ldots,C^{(m)}$.

To find a degree bound for the best circuit in $\calF_d$, our strategy is to optimize a \emph{weighted} degree instead. For a weight $\lambda>0$, define
\[
F^\star_d(\lambda)\coloneqq\min_{E\in\calF_d}\bigl\{\deg_{\sfL}(E)+\lambda\deg_{\sfR}(E)\bigr\}.
\]
This way we no longer need to control the two profiles $\deg_{\sfL}$ and $\deg_{\sfR}$ simultaneously, which is hard to keep track of when taking tensor products. Write $a^{(i)}_\zeta\coloneqq\deg_{\sfL}(C^{(i)}_\zeta)$ and $b^{(i)}_\zeta\coloneqq\deg_{\sfR}(C^{(i)}_\zeta)$; after discarding middle gates without edges, these are positive. For the circuit $E'$ constructed by~\eqref{eqn:family}, the same relaxation as in~\eqref{eqn:modest2} gives
\begin{align*}
\deg_{\sfL}(E')+\lambda\deg_{\sfR}(E')&\le\sum_{\zeta\in\calZ}\bigl(a^{(i)}_\zeta\deg_{\sfL}(E_\zeta)+\lambda b^{(i)}_\zeta\deg_{\sfR}(E_\zeta)\bigr)\\
&=\sum_{\zeta\in\calZ}a^{(i)}_\zeta\!\left(\deg_{\sfL}(E_\zeta)+\lambda\frac{b^{(i)}_\zeta}{a^{(i)}_\zeta}\deg_{\sfR}(E_\zeta)\right)\!.
\end{align*}
Now suppose we want to find the best $E'$, which comes from first determining some $i\in[m]$ and then the choice $E\colon\calZ\to\calF_{d-1}$. From the form above, the choice actually splits: each $E_\zeta$ can be optimized independently, with the rescaled weight $\lambda b^{(i)}_\zeta/a^{(i)}_\zeta$. Choosing every $E_\zeta$ optimally, we obtain a circuit $E'$ with
\[
\deg_{\sfL}(E')+\lambda\deg_{\sfR}(E')\le G^{(i)}_d(\lambda)\coloneqq\sum_{\zeta\in\calZ}a^{(i)}_\zeta\,F^\star_{d-1}\!\left(\lambda\frac{b^{(i)}_\zeta}{a^{(i)}_\zeta}\right)\!.
\]
So $G^{(i)}_d(\lambda)$ is the best bound we get when the outermost circuit is $C^{(i)}$. Since the family $\calF_d$ goes through all $C^{(i)}$, we have
\[
F^\star_d(\lambda)\le\min_{i\in[m]}G^{(i)}_d(\lambda).
\]

\subsubsection{Colored volume}

The value of $G^{(i)}_d(\lambda)$ depends on the values of $F^\star_{d-1}$ at other points, which makes this recursion seem hard to analyze. However, the behaviour of this functional recursion has been resolved in the work of Alman and Li~\cite{AL25}, which is the key input of this section.

We begin with some notation before stating their result. Notice that the recursion above depends only on the degrees of the color classes, and no longer on the specific details of the circuits. This allows the following simplification. A \emph{degree-pair list} is a finite nonempty list $W$ of pairs $(a,b)$ of reals with $a,b\ge1$; its \emph{volume} is $\rho_W(\tau)\coloneqq\sum_{(a,b)\in W}a^\tau b^{1-\tau}$ for $\tau\in[0,1]$. We call the one-element lists $[(A,1)]$ and $[(1,B)]$ with $A,B\ge1$ \emph{trivial}.

\begin{lemma}[Essentially {\cite[Theorem~B.2]{AL25}}]\label{lem:skew-recursion}
Let $W_1,\ldots,W_m$ be degree-pair lists, including a trivial list of each of the two forms. Define functions $F_d\colon(0,\infty)\to(0,\infty)$ by the recursion
\[
F_0(\lambda)\coloneqq1+\lambda,\qquad F_d(\lambda)\coloneqq\min_{1\le i\le m}\sum_{(a,b)\in W_i}a\,F_{d-1}\!\left(\lambda\frac{b}{a}\right)\qquad(d\ge1).
\]
Then we have the estimate
\[
F_d(\lambda)\le2^{o(d)}\max_{0\le\tau\le1}\rho(\tau)^d\lambda^{1-\tau},
\]
where the $o(d)$ term depends on $W_1,\ldots,W_m$ but not on $\lambda$, and $\rho$ is defined by
\[
\rho(\tau)\coloneqq\min_{1\le i\le m}\rho_{W_i}(\tau).
\]
\end{lemma}

We remark that the original statement in~\cite[Appendix~B]{AL25}, building on~\cite{JS13,AGP23,Ser22}, is more concrete: it is stated for optimizing the size of circuits. However, a look at the proof shows that it makes no use of special properties of circuit sizes, so it applies verbatim to abstract degree-pair lists as well.

To connect the lemma with the sketch above, take $W_i\coloneqq[(a^{(i)}_\zeta,b^{(i)}_\zeta)]_{\zeta\in\calZ}$. Since $F^\star_0(\lambda)=1+\lambda$ and the right-hand side of the recursion is monotone in $F_{d-1}$, induction on $d$ shows $F^\star_d\le F_d$. Thus the lemma bounds the weighted degree of the best rebalancing in terms of the volumes $\rho_{W_i}(\tau)$ alone. This motivates us to consider the \emph{colored volume} of circuits, which is exactly the volume of this degree-pair list. Let $\tau\in[0,1]$ be a parameter. The \emph{$\tau$-colored volume} of a colored circuit $C$ is defined by
\[
\rho_C(\tau)\coloneqq\sum_{\zeta\in\calZ}\deg_{\sfL}(C_\zeta)^\tau\deg_{\sfR}(C_\zeta)^{1-\tau}.
\]
Thus $\rho_C(1)$ and $\rho_C(0)$ are the sums of the left and of the right degrees of the color classes, and $\rho_C(\tau)$ interpolates between them. For a single color, $\rho_C(\tau)=\deg_{\sfL}(C)^\tau\deg_{\sfR}(C)^{1-\tau}$.

We now turn Lemma~\ref{lem:skew-recursion} back into a statement about circuits, which makes the sketch above precise.

\begin{corollary}\label{cor:rebalancing}
Let $C^{(1)},\ldots,C^{(m)}$ be colored $\EXACTZ$ circuits computing the same nonzero matrix $M$, among which there are a single-colored circuit with right degree $1$ and a single-colored circuit with left degree $1$, and put
\[
\rho(\tau)\coloneqq\min_{1\le i\le m}\rho_{C^{(i)}}(\tau).
\]
Then for every $d\ge0$ and every $\lambda>0$ there is an $\EXACTZ$ circuit $E_{d,\lambda}$ computing $M^{\otimes d}$ with
\[
\deg_{\sfL}(E_{d,\lambda})+\lambda\deg_{\sfR}(E_{d,\lambda})\le2^{o(d)}\max_{0\le\tau\le1}\rho(\tau)^d\lambda^{1-\tau},
\]
where the $o(d)$ term depends on $C^{(1)},\ldots,C^{(m)}$ but not on $\lambda$.
\end{corollary}

For example, the trivial circuits $(M,I)$ and $(I,M)$, each with a single color, can always be added to the list.

\begin{proof}[Proof of Corollary~\ref{cor:rebalancing}]
For each $i$, first delete middle gates with no incoming or no outgoing edge; this changes neither the computed matrix nor increases any degree or colored volume. Let $\calZ_i$ be the nonempty colors of $C^{(i)}$, and define $a^{(i)}_\zeta\coloneqq\deg_{\sfL}(C^{(i)}_\zeta)$ and $b^{(i)}_\zeta\coloneqq\deg_{\sfR}(C^{(i)}_\zeta)$. These degrees are positive integers, so $W_i\coloneqq[(a^{(i)}_\zeta,b^{(i)}_\zeta)]_{\zeta\in\calZ_i}$ is a degree-pair list with $\rho_{W_i}(\tau)=\rho_{C^{(i)}}(\tau)$ for all $\tau\in[0,1]$. Since $M\ne0$, the two single-colored circuits of the assumption give trivial lists of both forms. Let $F_d$ be as in Lemma~\ref{lem:skew-recursion} for $W_1,\ldots,W_m$; then its function $\rho$ is the one above, and it remains to construct, for every $d\ge0$ and $\lambda>0$, a circuit $E_{d,\lambda}$ computing $M^{\otimes d}$ with $\deg_{\sfL}(E_{d,\lambda})+\lambda\deg_{\sfR}(E_{d,\lambda})\le F_d(\lambda)$.

The recursion for $F_d$ tells us how to build such circuits: at each level, pick the best base circuit, and recurse on each of its color classes with the weight $\lambda$ rescaled by the ratio of the two degrees. We proceed by induction on $d$. At $d=0$, take a circuit with one middle gate and two edge coefficients equal to $1$. For $d\ge1$, choose at the root an index $i$ attaining the minimum in $F_d(\lambda)$. For every color $\zeta\in\calZ_i$, apply the induction hypothesis at $\lambda b^{(i)}_\zeta/a^{(i)}_\zeta$ to obtain $E_\zeta\coloneqq E_{d-1,\lambda b^{(i)}_\zeta/a^{(i)}_\zeta}$, and set
\[
E_{d,\lambda}\coloneqq\sum_{\zeta\in\calZ_i}C^{(i)}_\zeta\otimes E_\zeta,
\]
as in~\eqref{eqn:family}. This sum shares the input and output gates and combines the branches' disjoint middle gates. A middle gate in each tensor product has name $(e,f)$, where $e$ and $f$ name middle gates in its two factors; each edge coefficient is the product of the corresponding integer coefficients. The matrices computed by the $C^{(i)}_\zeta$ sum to $M$, and each $E_\zeta$ computes $M^{\otimes(d-1)}$. Hence $E_{d,\lambda}$ computes $M^{\otimes d}$. Tensor products multiply the two degrees, and combining the branches gives each degree at most the sum of the branch degrees. The induction hypothesis therefore bounds $\deg_{\sfL}(E_{d,\lambda})+\lambda\deg_{\sfR}(E_{d,\lambda})$ by $\sum_{\zeta\in\calZ_i}a^{(i)}_\zeta F_{d-1}(\lambda b^{(i)}_\zeta/a^{(i)}_\zeta)=F_d(\lambda)$.
\end{proof}

The following lemma is the form in which we use rebalancing: a bound on the colored volume of circuits for $P_{k,1/2}$ yields points of $\calK(1/2)$.

\begin{lemma}\label{lem:colored-transfer}
Let $V\colon[0,1]\to\mathbb{R}_{\ge0}$, and suppose that for every even $k\ge2$ there is a colored $\EXACTZ$ circuit $C_k$ computing $P_{k,1/2}$ such that
\[
\rho_{C_k}(\tau)\le2^{(V(\tau)+o(1))k}\qquad\text{uniformly in }\tau\in[0,1].
\]
Then every pair $(l,r)\in\mathbb{R}^2$ with
\[
\tau l+(1-\tau)r\ge\min\left\{V(\tau),\tau,\frac{1-\tau}{2}\right\}\qquad\text{for all }\tau\in[0,1]
\]
belongs to $\calK(1/2)$.
\end{lemma}

In the condition on $(l,r)$, the terms $\tau$ and $(1-\tau)/2$ come from the two trivial circuits $(P_{k,1/2},I)$ and $(I,P_{k,1/2})$, with degrees $\bigl(\binom{k}{k/2},1\bigr)$ and $(1,2^{k/2})$; rebalancing lets us use whichever of the three circuits is best at each $\tau$.

\begin{proof}[Proof of Lemma~\ref{lem:colored-transfer}]
The plan is as follows. We apply Corollary~\ref{cor:rebalancing} to the given circuit and the two trivial circuits to obtain circuits for the tensor powers $(P_{k,1/2})^{\otimes d}$, and then pass to $P_{n,1/2}$ by Lemma~\ref{lem:block-grouping}.

The assumptions at $\tau=1$ and $\tau=0$ give $l,r\ge0$. Fix $\varepsilon>0$. By the definition of $\calK(1/2)$, it suffices to construct, for every sufficiently large $n$, a circuit computing $P_{n,1/2}$ with degrees $\le2^{(l+\varepsilon)n}$ and $\le2^{(r+\varepsilon)n}$. By the hypothesis, we can choose and fix an even $k$ with $\h(1/k)\le\varepsilon/4$ such that $C^{(1)}\coloneqq C_k$ satisfies $\rho_{C^{(1)}}(\tau)\le2^{(V(\tau)+\varepsilon/4)k}$ for every $\tau\in[0,1]$.

Add the two trivial circuits $C^{(2)}\coloneqq(P_{k,1/2},I)$ and $C^{(3)}\coloneqq(I,P_{k,1/2})$, each colored with a single color. Their colored volumes are $\rho_{C^{(2)}}(\tau)=\binom{k}{k/2}^\tau\le2^{\tau k}$ and $\rho_{C^{(3)}}(\tau)=2^{(1-\tau)k/2}$. Together with the assumption on $(l,r)$, this gives $\rho(\tau)\le2^{(\tau l+(1-\tau)r+\varepsilon/4)k}$ for the function $\rho$ of Corollary~\ref{cor:rebalancing}.

For every even $n\ge4$, put $d\coloneqq\lceil n/k\rceil$ and $\lambda\coloneqq2^{(l-r)kd}$. Then $\rho(\tau)^d\lambda^{1-\tau}\le2^{(l+\varepsilon/4)kd}$ for every $\tau$, so Corollary~\ref{cor:rebalancing} gives a circuit $E_{d,\lambda}$ computing $(P_{k,1/2})^{\otimes d}$ with $\deg_{\sfL}(E_{d,\lambda})\le2^{(l+\varepsilon/4)kd+o(d)}$ and $\deg_{\sfR}(E_{d,\lambda})\le\lambda^{-1}2^{(l+\varepsilon/4)kd+o(d)}=2^{(r+\varepsilon/4)kd+o(d)}$. Applying Lemma~\ref{lem:block-grouping} to $E_{d,\lambda}$ gives a circuit $C_n$ computing $P_{n,1/2}$ with $\deg_{\sfL}(C_n)\le2^{(l+\varepsilon/4+\h(1/k))kd+o(d)}$ and $\deg_{\sfR}(C_n)\le2^{(r+\varepsilon/4)kd+o(d)}$.

Since $k$ is fixed, $kd-n<k$ and $o(d)=o(n)$. Choose an integer $n_0\ge4$ so that, for every integer $n\ge n_0$ and $d=\lceil n/k\rceil$,
\[
o(d)+(\max\{l,r\}+\varepsilon)(kd-n)\le\frac{\varepsilon n}{4}.
\]
For every even $n\ge n_0$, the two degree upper bounds are then $2^{(l+3\varepsilon/4)n}$ and $2^{(r+\varepsilon/2)n}$. For odd $n$, use the embedding into length $n+1$ from Section~\ref{sec:circuit-prelim}. Because $k$ is even, $\lceil(n+1)/k\rceil=\lceil n/k\rceil$, so the same degree upper bounds in terms of $kd$ and the same threshold condition apply. Since $\varepsilon>0$ was arbitrary, $(l,r)\in\calK(1/2)$.
\end{proof}

\subsection{The circuit construction} \label{sec:integer-circuit}

In this subsection, we construct an $\EXACTZ$ circuit achieving $l + r < 1/2$.

We first construct a circuit for a structured set of columns, and then use the symmetries of $P_{k,1/2}$ and a random covering to handle all queries. Throughout this subsection, $k$ is even and $a,b>0$ are such that $ak$ and $bk$ are integers with $4ak\le k-bk$; in the proof of Theorem~\ref{thm:half-stars-circuit} we reduce to this case by rounding.

All exponents below come from counting with binomial coefficients. Recall the function $\h(u,v)=u\log_2u-v\log_2v-(u-v)\log_2(u-v)$ from Section~\ref{sec:circuit-prelim}. Since $\h(u,v)=u\h(v/u)$, the standard estimate $2^{n\h(m/n)}/(n+1)\le\binom nm\le2^{n\h(m/n)}$ with $n=uk$ and $m=vk$ becomes
\begin{equation}\label{eqn:binom-entropy}
\frac{2^{\h(u,v)k}}{uk+1}\le\binom{uk}{vk}\le2^{\h(u,v)k}\qquad\text{whenever }uk,vk\text{ are integers}.
\end{equation}
In particular, $\binom{uk}{vk}=2^{\h(u,v)k+O(\log k)}$ and $\binom{2m}{m}\ge2^{2m}/(2m+1)$.

For $B\subseteq[k]$, let $\calY_B$ consist of the queries $\bfy$ with $\lvert S(\bfy)\rvert=k/2$ whose specified bits outside $B$ are balanced: $S(\bfy)\setminus B$ contains exactly $ak$ positions $i$ with $y_i=1$ and exactly $ak$ positions $i$ with $y_i=0$.

\begin{lemma}\label{lem:fixed-B-circuit}
Let $B\subseteq[k]$ with $\lvert B\rvert=bk$. Then the restriction of $P_{k,1/2}$ to the columns $\calY_B$ has a colored $\EXACTZ$ circuit $C_B$, whose colors are ordered pairs $(\mu,\theta)$ satisfying $a\le\mu\le2a$, $\mu\le\theta\le2\mu$, and $\mu k,\theta k\in\mathbb{Z}$.

Moreover, for every color $(\mu,\theta)$ and every binary string $\bfx$, writing $C'\coloneqq(C_B)_{(\mu,\theta)}$, we have
\[
\begin{aligned}
\deg_{\sfL}(C')&\le O^*\bigl(2^{(\h(1-b,\mu)-2\mu)k}\bigr),\\
\mathbb{E}_{\bfz}\bigl[\deg_{\sfL}(C',\bfx\oplus\bfz)\bigr]&\le O^*\bigl(2^{(\h(1-b,\mu)+\h(2\mu,\theta)-4\mu)k}\bigr),\\
\deg_{\sfR}(C')&\le2^{(b-1/2+2a+\h(a,\mu-a)+\h(\mu,\theta-\mu))k},
\end{aligned}
\]
where $\bfz\in\{0,1\}^k$ is uniformly random subject to $z_i=0$ for all $i\in B$, and $\oplus$ is the bitwise XOR. Here $O^*(\cdot)$ hides only factors polynomial in $k$, uniformly in $B$ and the color.
\end{lemma}

\begin{proof}[Proof of Lemma~\ref{lem:fixed-B-circuit}]
The circuit treats the positions inside and outside $B$ differently. Inside $B$, it simply enumerates the assignments of $\bfx$. Outside $B$, it expands the constraints by inclusion--exclusion and evaluates each resulting monomial using a covering family.

For $\bfy\in\calY_B$ and $i\in\{0,1\}$, let $A_i(\bfy)\coloneqq\{j\in S(\bfy)\setminus B:y_j=i\}$, so that $\lvert A_0(\bfy)\rvert=\lvert A_1(\bfy)\rvert=ak$. Expanding $\prod_{j\in A_0(\bfy)}(1-x_j)$ gives
\[
f_{\bfy}(\bfx)=f_{\bfy|_B}(\bfx|_B)\sum_{A_1(\bfy)\subseteq U\subseteq A_1(\bfy)\cup A_0(\bfy)}(-1)^{\lvert U\rvert-ak}\prod_{i\in U}x_i.
\]
Every set $U$ in this sum satisfies $U\subseteq\overline B$ and $ak\le\lvert U\rvert\le2ak$. For every integer $\mu k$ in this range, $2\mu k\le4ak\le k-bk$, we let $\mathscr{Q}_{\overline B,\mu}$ be a family of $2\mu k$-subsets of $\overline B$ that covers all $\mu k$-subsets of $\overline B$. By a classical probabilistic argument, such a family exists with
\[
\lvert\mathscr{Q}_{\overline B,\mu}\rvert\le k\frac{\binom{k-bk}{\mu k}}{\binom{2\mu k}{\mu k}}.
\]
Indeed, a uniformly random $2\mu k$-subset of $\overline B$ contains a fixed $\mu k$-subset with probability $\binom{2\mu k}{\mu k}/\binom{k-bk}{\mu k}$, so with this many independent random sets, each of the at most $2^k$ subsets is missed with probability at most $e^{-k}$, and a union bound applies. By~\eqref{eqn:binom-entropy}, $\binom{k-bk}{\mu k}\le2^{\h(1-b,\mu)k}$, and $\binom{2\mu k}{\mu k}\ge2^{2\mu k}/(2\mu k+1)$, so
\[
\lvert\mathscr{Q}_{\overline B,\mu}\rvert\le k(2\mu k+1)\,2^{(\h(1-b,\mu)-2\mu)k}.
\]
Assign each $\mu k$-subset $U\subseteq\overline B$ to exactly one containing set $Q(U)\in\mathscr{Q}_{\overline B,\mu}$.

The middle gates have names $(Q,\boldsymbol{\sigma})$, where $Q\in\mathscr{Q}_{\overline B,\mu}$ for some $\mu$ and $\boldsymbol{\sigma}\in\{0,1\}^{B\cup Q}$ is an assignment to the positions in $B\cup Q$. Set
\[
\begin{aligned}
(\bfL_{C_B})_{\bfx,(Q,\boldsymbol{\sigma})}&\coloneqq\bfone[\bfx|_{B\cup Q}=\boldsymbol{\sigma}],\\
(\bfR_{C_B})_{(Q,\boldsymbol{\sigma}),\bfy}&\coloneqq f_{\bfy|_B}(\boldsymbol{\sigma}|_B)\sum_{\substack{Q(U)=Q\\A_1(\bfy)\subseteq U\subseteq A_1(\bfy)\cup A_0(\bfy)}}(-1)^{\lvert U\rvert-ak}\prod_{i\in U}\sigma_i.
\end{aligned}
\]
For each $Q$, the binary string $\bfx$ is adjacent to exactly one middle gate $(Q,\bfx|_{B\cup Q})$. Since each $U$ is assigned to exactly one $Q$, the expansion above gives $(\bfL_{C_B}\bfR_{C_B})_{\bfx,\bfy}=f_{\bfy}(\bfx)$ over the integers for every binary string $\bfx$ and every $\bfy\in\calY_B$.

Give the middle gate $(Q,\boldsymbol{\sigma})$ the color $(\mu,\theta)$, where $\mu\coloneqq\lvert Q\rvert/(2k)$ and $\theta\coloneqq\#1(\boldsymbol{\sigma}|_Q)/k$, and delete middle gates with no incoming or no outgoing edge. A nonzero right coefficient requires some $U\subseteq Q$ with $\lvert U\rvert=\mu k$ and $\boldsymbol{\sigma}|_U\equiv1$, so the remaining colors satisfy $\mu\le\theta\le2\mu$.

For the left degree, the middle gates of color $(\mu,\theta)$ adjacent to $\bfx$ are among the gates $(Q,\bfx|_{B\cup Q})$ with $Q\in\mathscr{Q}_{\overline B,\mu}$ and $\#1(\bfx|_Q)=\theta k$. Hence $\deg_{\sfL}(C',\bfx)\le\lvert\mathscr{Q}_{\overline B,\mu}\rvert$, which gives the first bound. For the second, every $Q\in\mathscr{Q}_{\overline B,\mu}$ lies in $\overline B$, so $(\bfx\oplus\bfz)|_Q$ is uniform on $\{0,1\}^Q$, and
\[
\mathbb{E}_{\bfz}\bigl[\deg_{\sfL}(C',\bfx\oplus\bfz)\bigr]\le\lvert\mathscr{Q}_{\overline B,\mu}\rvert\,2^{-2\mu k}\binom{2\mu k}{\theta k}\le\lvert\mathscr{Q}_{\overline B,\mu}\rvert\,2^{(\h(2\mu,\theta)-2\mu)k},
\]
using $\binom{2\mu k}{\theta k}\le2^{\h(2\mu,\theta)k}$ from~\eqref{eqn:binom-entropy}. Together with the bound on $\lvert\mathscr{Q}_{\overline B,\mu}\rvert$, the exponent is $\h(1-b,\mu)-2\mu+\h(2\mu,\theta)-2\mu$, which gives the second bound.

For the right degree, fix $\bfy\in\calY_B$. A nonzero right coefficient requires $\boldsymbol{\sigma}|_B$ to agree with $\bfy$ on $S(\bfy)\cap B$, which has $k/2-2ak$ positions; this leaves $2^{bk-k/2+2ak}$ assignments on $B$. It also requires some $U$ with $Q(U)=Q$ and $\boldsymbol{\sigma}|_U\equiv1$. There are $\binom{ak}{\mu k-ak}$ choices of $U$ of size $\mu k$, each of which determines $Q$ and $\boldsymbol{\sigma}|_U$. Then $\binom{\mu k}{\theta k-\mu k}$ assignments on $Q\setminus U$ give the color $(\mu,\theta)$. Counting a middle gate more than once gives the upper bound
\[
2^{bk-k/2+2ak}\binom{ak}{\mu k-ak}\binom{\mu k}{\theta k-\mu k}\le2^{(b-1/2+2a+\h(a,\mu-a)+\h(\mu,\theta-\mu))k},
\]
where the three factors contribute the exponents $b-1/2+2a$, $\h(a,\mu-a)$, and $\h(\mu,\theta-\mu)$, respectively, by~\eqref{eqn:binom-entropy}.
\end{proof}

We now assemble the circuits of Lemma~\ref{lem:fixed-B-circuit} into a single colored circuit for the whole matrix $P_{k,1/2}$: flipping bits reaches the unbalanced queries, and a random family of pairs $(B,\boldsymbol{\varphi})$ covers all of them. The following theorem bounds the two degrees of every color class of the resulting circuit.

\begin{theorem}\label{thm:half-stars-circuit}
Let $a,b\in\mathbb{R}$ satisfy $a>0$ and $4a\le1-b\le1/2$, and put
\[
\kappa\coloneqq\h(b)-\h(1/2,2a)-\h(1/2,1-b-2a).
\]
Assume that $\kappa>4a$. For $a\le\mu\le2a$ and $\mu\le\theta\le2\mu$, define
\[
\begin{aligned}
V_{\sfL}(\mu,\theta)&\coloneqq\kappa+\h(1-b,\mu)+\h(2\mu,\theta)-4\mu,\\
V_{\sfR}(\mu,\theta)&\coloneqq b-\tfrac12+2a+\h(a,\mu-a)+\h(\mu,\theta-\mu).
\end{aligned}
\]
Then for every sufficiently large even $k$, there is a colored $\EXACTZ$ circuit $C_k$ computing $P_{k,1/2}$, with at most $(k+1)^2$ colors, each of which is a pair $(\mu,\theta)$ with $a\le\mu\le2a$ and $\mu\le\theta\le2\mu$.

Moreover, for every color $(\mu,\theta)$, writing $C'\coloneqq(C_k)_{(\mu,\theta)}$, we have
\[
\deg_{\sfL}(C')\le2^{(V_{\sfL}(\mu,\theta)+o(1))k},\qquad\deg_{\sfR}(C')\le2^{(V_{\sfR}(\mu,\theta)+o(1))k},
\]
where the $o(1)$ terms are uniform over all colors.
\end{theorem}

\begin{proof}[Proof of Theorem~\ref{thm:half-stars-circuit}]
We construct the required colored circuits for every sufficiently large even $k$. To simplify notation, we write $ak$ and $bk$ for $\lfloor ak\rfloor$ and $\lfloor bk\rfloor$ throughout the proof; that is, we argue as if $ak$ and $bk$ were integers. The effect of this rounding is handled in the last paragraph.

\paragraph{Flipping symmetry.}
Lemma~\ref{lem:fixed-B-circuit} handles only the queries whose specified bits outside $B$ are balanced: $ak$ ones and $ak$ zeros. To reach the other queries, we use the symmetry of $P_1$ that swaps the roles of $0$ and $1$. For $\boldsymbol{\varphi}\in\{0,1\}^k$, write $\bfx\oplus\boldsymbol{\varphi}$ for the bitwise XOR, and let $\bfy\oplus\boldsymbol{\varphi}$ be the query obtained from $\bfy$ by replacing $y_i$ with $y_i\oplus\varphi_i$ for every $i\in S(\bfy)$ and keeping the wildcards. Then $f_{\bfy\oplus\boldsymbol{\varphi}}(\bfx\oplus\boldsymbol{\varphi})=f_{\bfy}(\bfx)$, and $\lvert S(\bfy\oplus\boldsymbol{\varphi})\rvert=\lvert S(\bfy)\rvert$.

For a $bk$-subset $B$ and a vector $\boldsymbol{\varphi}\in\{0,1\}^{\overline B}$, extended by zeros on $B$, let $\calY_{B,\boldsymbol{\varphi}}\coloneqq\{\bfy:\bfy\oplus\boldsymbol{\varphi}\in\calY_B\}$. Let $C_B$ be the circuit of Lemma~\ref{lem:fixed-B-circuit}, and define $C_{B,\boldsymbol{\varphi}}$ by relabeling its input and output gates:
\[
(\bfL_{C_{B,\boldsymbol{\varphi}}})_{\bfx,e}\coloneqq(\bfL_{C_B})_{\bfx\oplus\boldsymbol{\varphi},e},\qquad(\bfR_{C_{B,\boldsymbol{\varphi}}})_{e,\bfy}\coloneqq(\bfR_{C_B})_{e,\bfy\oplus\boldsymbol{\varphi}}\qquad(\bfy\in\calY_{B,\boldsymbol{\varphi}}),
\]
with the same coloring. By the symmetry, $C_{B,\boldsymbol{\varphi}}$ is an $\EXACTZ$ circuit computing the restriction of $P_{k,1/2}$ to the columns $\calY_{B,\boldsymbol{\varphi}}$. It has the same right degrees as $C_B$, and the left degree of $\bfx$ in $C_{B,\boldsymbol{\varphi}}$ equals the left degree of $\bfx\oplus\boldsymbol{\varphi}$ in $C_B$, color by color.

\paragraph{Covering all queries.}
We choose a family $\mathscr{B}$ of pairs $(B,\boldsymbol{\varphi})$ such that the sets $\calY_{B,\boldsymbol{\varphi}}$ cover all queries. Choose $B$ uniformly among the $bk$-subsets of $[k]$, and then choose the bits of $\boldsymbol{\varphi}\in\{0,1\}^{\overline B}$ independently and fairly. A fixed query $\bfy$ belongs to $\calY_{B,\boldsymbol{\varphi}}$ if and only if $\overline B$ contains exactly $2ak$ specified positions of $\bfy$ and the flipped bits $y_i\oplus\varphi_i$ on these positions contain exactly $ak$ ones. This happens with probability
\[
p_0\coloneqq\frac{\binom{k/2}{2ak}\binom{k/2}{k-bk-2ak}}{\binom{k}{bk}}
\frac{\binom{2ak}{ak}}{2^{2ak}},
\]
which is the same for all queries. Here the first factor is the hypergeometric probability that exactly $2ak$ of the $k/2$ specified positions of $\bfy$ fall into the $(k-bk)$-set $\overline B$, and the second is the probability that $2ak$ fair bits contain exactly $ak$ ones. We estimate each binomial coefficient by~\eqref{eqn:binom-entropy}:
\[
\binom{k/2}{2ak}=2^{\h(1/2,2a)k+O(\log k)},\qquad\binom{k/2}{k-bk-2ak}=2^{\h(1/2,1-b-2a)k+O(\log k)},
\]
\[
\binom{k}{bk}=2^{\h(b)k+O(\log k)},\qquad\frac{\binom{2ak}{ak}}{2^{2ak}}=2^{O(\log k)}.
\]
Hence, with $\kappa=\h(b)-\h(1/2,2a)-\h(1/2,1-b-2a)$ as in the theorem,
\begin{equation}\label{eqn:p0-estimate}
p_0=2^{-\kappa k+O(\log k)}.
\end{equation}
Since $\kappa>4a$, this gives $p_0\le2^{-4ak}$ for all sufficiently large $k$; this comparison will be used for the left degrees.

Take $\lceil k^2/p_0\rceil$ independent samples. A union bound over the at most $3^k$ queries gives probability $\le3^ke^{-k^2}$ that some query is not covered. For any samples covering all queries, remove repeated pairs to obtain $\mathscr{B}$, and assign each query $\bfy$ exactly one pair $(B(\bfy),\boldsymbol{\varphi}(\bfy))\in\mathscr{B}$ with $\bfy\in\calY_{B(\bfy),\boldsymbol{\varphi}(\bfy)}$. Below we will choose samples for which the left-degree upper bounds also hold.

\paragraph{The circuit.}
For each $(B,\boldsymbol{\varphi})\in\mathscr{B}$, keep in $C_{B,\boldsymbol{\varphi}}$ only the output gates of the queries assigned to $(B,\boldsymbol{\varphi})$. Let $C$ be the sum of these circuits: it shares the input and output gates and combines their disjoint middle gates, whose names are $(B,\boldsymbol{\varphi},Q,\boldsymbol{\sigma})$; the coloring is inherited. Every query $\bfy$ receives edges only from $C_{B(\bfy),\boldsymbol{\varphi}(\bfy)}$, so $(\bfL_C\bfR_C)_{\bfx,\bfy}=f_{\bfy}(\bfx)$ over the integers, and $C$ computes $P_{k,1/2}$. Delete middle gates with no incoming or no outgoing edge. For every color $(\mu,\theta)$, the right degree of each query is bounded as in Lemma~\ref{lem:fixed-B-circuit}:
\[
\deg_{\sfR}(C_{(\mu,\theta)})\le2^{(b-1/2+2a+\h(a,\mu-a)+\h(\mu,\theta-\mu))k}.
\]

\paragraph{The left degree.}
Fix a binary string $\bfx$ and a color $(\mu,\theta)$. Before removing repeated pairs and restricting the output gates, the left degree of $\bfx$ in $C_{(\mu,\theta)}$ is at most $\sum_jD_j$, where for the $j$-th sample $(B_j,\boldsymbol{\varphi}_j)$ we put $D_j\coloneqq\deg_{\sfL}\bigl((C_{B_j})_{(\mu,\theta)},\bfx\oplus\boldsymbol{\varphi}_j\bigr)$. Let $N_\mu=O^*\bigl(2^{(\h(1-b,\mu)-2\mu)k}\bigr)$ and $q_{\mu,\theta}=O^*\bigl(2^{(\h(1-b,\mu)+\h(2\mu,\theta)-4\mu)k}\bigr)$ denote the first two bounds of Lemma~\ref{lem:fixed-B-circuit}; they do not depend on $B$. Conditional on $B_j$, the vector $\boldsymbol{\varphi}_j$ is distributed as $\bfz$ in the lemma. Hence the variables $D_j/N_\mu$ are independent, lie in $[0,1]$, and have expectation at most $q_{\mu,\theta}/N_\mu$.

Put $M\coloneqq\lceil k^2/p_0\rceil q_{\mu,\theta}/N_\mu+k^2$, which is at least the expectation of $\sum_jD_j/N_\mu$. The multiplicative Chernoff bound gives probability $\le e^{-M/3}\le e^{-k^2/3}$ that this sum exceeds $2M$. After a union bound over all $2^k$ binary strings and $\le(k+1)^2$ colors, this failure probability still tends to $0$, as does the probability of an uncovered query. For all sufficiently large $k$, the combined failure probability is $<1$, so there are samples for which every query is covered and all the upper bounds
\[
\deg_{\sfL}(C_{(\mu,\theta)})\le2N_\mu M=2\left\lceil\frac{k^2}{p_0}\right\rceil q_{\mu,\theta}+2k^2N_\mu\le\frac{1}{p_0}\,O^*\bigl(2^{(\h(1-b,\mu)+\h(2\mu,\theta)-4\mu)k}\bigr)
\]
hold simultaneously; for the last inequality, the second term $2k^2N_\mu=O^*\bigl(2^{(\h(1-b,\mu)-2\mu)k}\bigr)$ is absorbed into the first: since $1/p_0\ge2^{4ak}\ge2^{2\mu k}$ and $\h(2\mu,\theta)\ge0$, we have $2^{(\h(1-b,\mu)-2\mu)k}\le\frac{1}{p_0}2^{(\h(1-b,\mu)+\h(2\mu,\theta)-4\mu)k}$. Removing repeated pairs and middle gates only decreases these counts. Fix such a family $\mathscr{B}$ and the query assignments.

\paragraph{Degree bounds.}
By~\eqref{eqn:p0-estimate}, $1/p_0=2^{\kappa k+O(\log k)}$. Substituting this into the left-degree bound and comparing with the definitions of $V_{\sfL}$ and $V_{\sfR}$ gives
\begin{align*}
\deg_{\sfL}(C_{(\mu,\theta)})&\le2^{(\kappa+\h(1-b,\mu)+\h(2\mu,\theta)-4\mu)k+O(\log k)}=2^{V_{\sfL}(\mu,\theta)k+O(\log k)},\\
\deg_{\sfR}(C_{(\mu,\theta)})&\le2^{V_{\sfR}(\mu,\theta)k},
\end{align*}
with $V_{\sfL},V_{\sfR}$ evaluated at the rounded parameters $a_k\coloneqq\lfloor ak\rfloor/k$ and $b_k\coloneqq\lfloor bk\rfloor/k$ in place of $a$ and $b$. The colors $(\mu,\theta)$ satisfy $\lfloor ak\rfloor\le\mu k\le2\lfloor ak\rfloor$ and $\mu k\le\theta k\le2\mu k$ with $\mu k,\theta k$ integers, so there are at most $(k+1)^2$ of them. To place them in the domain of the theorem, rename each color $(\mu,\theta)$ to $(\mu,\theta)\cdot a/a_k$, a rescaling by $1+O(1/k)$ that preserves $\mu\le\theta\le2\mu$. Uniform continuity of the generalized entropy absorbs the floors and this rescaling into an $o(1)$ error in the exponents, uniform over colors. Hence $C_k\coloneqq C$ is the required circuit.
\end{proof}

Summing the degree bounds of Theorem~\ref{thm:half-stars-circuit} over the colors bounds the colored volume of $C_k$. Feeding this into the rebalancing Lemma~\ref{lem:colored-transfer}, we obtain points of $\calK(1/2)$, which is the form used in Section~\ref{sec:together}.

\begin{corollary}\label{cor:half-stars-K}
Let $a,b$ and $V_{\sfL},V_{\sfR}$ be as in Theorem~\ref{thm:half-stars-circuit}, and define
\[
V(\tau)\coloneqq\max_{\substack{a\le\mu\le2a\\\mu\le\theta\le2\mu}}\bigl\{\tau V_{\sfL}(\mu,\theta)+(1-\tau)V_{\sfR}(\mu,\theta)\bigr\}\qquad(\tau\in[0,1]).
\]
Then $V$ is nonnegative on $[0,1]$, and every pair $(l,r)\in\mathbb{R}^2$ with
\[
\tau l+(1-\tau)r\ge\min\left\{V(\tau),\tau,\frac{1-\tau}{2}\right\}\qquad\text{for all }\tau\in[0,1]
\]
belongs to $\calK(1/2)$.
\end{corollary}

\begin{proof}[Proof of Corollary~\ref{cor:half-stars-K}]
Let $C_k$ be the circuits of Theorem~\ref{thm:half-stars-circuit}. Summing over its at most $(k+1)^2$ colors,
\begin{align*}
\rho_{C_k}(\tau)&=\sum_{(\mu,\theta)}\deg_{\sfL}\bigl((C_k)_{(\mu,\theta)}\bigr)^\tau\deg_{\sfR}\bigl((C_k)_{(\mu,\theta)}\bigr)^{1-\tau}\\
&\le(k+1)^2\max_{(\mu,\theta)}2^{(\tau V_{\sfL}(\mu,\theta)+(1-\tau)V_{\sfR}(\mu,\theta)+o(1))k}\le2^{(V(\tau)+o(1))k}
\end{align*}
uniformly in $\tau\in[0,1]$. Since $P_{k,1/2}$ is nonzero, some middle gate of $C_k$ has both an incoming and an outgoing edge, so $\rho_{C_k}(\tau)\ge1$; letting $k\to\infty$ gives $V(\tau)\ge0$. For the finitely many remaining even $k$, any colored circuit computing $P_{k,1/2}$ will do, since finitely many values of $k$ do not affect the $o(1)$ term. The hypothesis of Lemma~\ref{lem:colored-transfer} is therefore satisfied, and that lemma gives the claim.
\end{proof}

\subsection{A lower bound for \texorpdfstring{$\OR$}{OR} circuits}\label{sec:partial-match-circuits}

We close the section by showing that the improvement of Section~\ref{sec:integer-circuit} cannot be achieved without algebraic cancellation. Recall that for $P_{n,1/2}$ the trivial circuit $(I,P_{n,1/2})$, in which each query enumerates all $2^{n/2}$ completions of its wildcards, is a $\PAR$ circuit with $(l,r)=(0,1/2)$, so $l+r=1/2$. The following theorem shows that this is optimal even among $\OR$ circuits, whereas Corollary~\ref{cor:half-stars-K} gives $\EXACTZ$ circuits with $l+r<1/2$.

\begin{theorem}\label{thm:or-lower-bound}
Let $l,r\ge0$, and suppose that there are $\OR$ circuits $C_n$ computing $P_{n,1/2}$ with $\deg_{\sfL}(C_n)\le2^{(l+o(1))n}$ and $\deg_{\sfR}(C_n)\le2^{(r+o(1))n}$. Then $l+r\ge1/2$.
\end{theorem}

We remark that Burkhard, Fredman, and Kleitman~\cite{BFK81} and Fredman and Volper~\cite{FV82} proved this lower bound for $\PAR$ circuits. Their argument adapts directly to $\OR$ circuits; we include the proof for completeness.

\begin{proof}[Proof of Theorem~\ref{thm:or-lower-bound}]
Let $C$ be any $\OR$ circuit computing $P_{n,1/2}$, and fix $1/2\le b\le1$. Each middle gate of $C$ defines the all-one rectangle whose rows and columns correspond to its adjacent input and output gates. Split the middle gates into two classes according to whether these rectangles contain at least $2^{n-\lfloor bn\rfloor}$ rows or fewer than this threshold.
\begin{itemize}
\item In the first class, at most $\lfloor bn\rfloor$ coordinates are constant across all binary strings $\bfx$ corresponding to the rows of a rectangle, since $c$ constant coordinates leave at most $2^{n-c}$ strings. Every query in that rectangle must specify only such coordinates, with the constant values, so the rectangle has at most $\binom{\lfloor bn\rfloor}{n/2}$ columns. Each row belongs to at most $\deg_{\sfL}(C)$ rectangles of this class.
\item In the second class, each rectangle has fewer than $2^{n-\lfloor bn\rfloor}$ rows, and each column belongs to at most $\deg_{\sfR}(C)$ such rectangles.
\end{itemize}
The matrix has $2^n$ rows, $\binom n{n/2}2^{n/2}$ columns, and $\binom n{n/2}2^n$ ones. Bound the first class's total area by summing over rows, and the second class's by summing over columns. Their rectangles cover all the ones, so dividing by the number of ones gives
\[
\frac{\deg_{\sfL}(C)}{\binom n{n/2}\big/\binom{\lfloor bn\rfloor}{n/2}}+\frac{\deg_{\sfR}(C)}{2^{\lfloor bn\rfloor-n/2}}\ge1.
\]

If $r\ge1/2$, then $l+r\ge1/2$ trivially. Otherwise, take $b>1/2+r$, so that the second term tends to $0$. Since $\binom n{n/2}=2^{n+O(\log n)}$ and $\binom{\lfloor bn\rfloor}{n/2}\le2^{\h(b,1/2)n}$, the first term forces $l\ge1-\h(b,1/2)$, and letting $b\downarrow1/2+r$ gives
\[
l+r\ge g(r)\coloneqq1+r-\h(1/2+r,1/2).
\]
On $[0,1/2]$, we have $g'(r)=1-\log_2\frac{1/2+r}{r}\le0$, so $g(r)\ge g(1/2)=1-\h(1,1/2)+1/2=1/2$.
\end{proof}

\section{Putting everything together}\label{sec:together}\label{parameter-selection-numerical-computation-and-certification}\label{sec:numerical-certificate}

\begin{theorem}\label{thm:main}
  Subset sum can be solved by a randomized algorithm in time $O(2^{0.49999\overline{7}n})$.
\end{theorem}

The proof chains together three results from the previous sections:
\begin{itemize}
\item Corollary~\ref{cor:half-stars-K} turns the function $V$ into pairs $(l,r)\in\calK(1/2)$;
\item Corollary~\ref{cor:circuit-to-pm} shows that every such pair is an admissible exponent pair;
\item Theorem~\ref{thm:subset-sum-transfer} turns an admissible pair with $l+r=1/2-2\xi$ into a subset sum algorithm in time $O(2^{(1/2-\xi/9+\varepsilon)n})$.
\end{itemize}
It remains to choose $a,b$, compute $V$, pick $(l,r)$, and check numerically that the resulting $\xi$ is large enough. Every quantity below has an explicit closed form, and its numerical value is listed in Table~\ref{tab:numerics}.

\paragraph{Parameters.}
We take
\[
a\coloneqq0.00037954865735953624129927,\qquad b\coloneqq0.99094548027616196930017929,
\]
obtained by numerical optimization. Clearly $a>0$, $4a\le1-b\le1/2$, and $1-b>10a$; the table shows $\kappa>4a$. So Theorem~\ref{thm:half-stars-circuit} and Corollary~\ref{cor:half-stars-K} apply.

\paragraph{Three values of $V$.}
As a maximum of affine functions of $\tau$, the function $V$ is convex. We therefore only compute $V(0)$, $V(1/2)$, and $V(1)$, and bound $V$ on $[0,1]$ by the piecewise linear function $\widehat V$ through these three values.
\begin{itemize}
\item At $\tau=0$, we maximize $\h(a,\mu-a)+\h(\mu,\theta-\mu)$. The maximum over $\theta$ is $\mu$, at $\theta=3\mu/2$, and then the maximum of $\h(a,\mu-a)+\mu$ is $a\log_26$, at $\mu=5a/3$. Hence
\[
V(0)=b-\frac12+(2+\log_26)a.
\]
\item At $\tau=1$, we maximize $\h(1-b,\mu)+\h(2\mu,\theta)-4\mu$. The maximum over $\theta$ is at $\theta=\mu$, leaving $\h(1-b,\mu)-2\mu$. This is increasing for $\mu\le(1-b)/5$, hence on $[a,2a]$ because $1-b>10a$. Taking $\mu=2a$ gives
\[
V(1)=\kappa+\h(1-b,2a)-4a.
\]
\item At $\tau=1/2$, the maximum over $\theta$ of $\h(2\mu,\theta)+\h(\mu,\theta-\mu)$ is at $\theta^*\coloneqq4\mu/3$, where this sum minus $4\mu$ equals $\mu\log_2(27/64)$. The remaining objective $\h(1-b,\mu)+\h(a,\mu-a)+\mu\log_2(27/64)$ is strictly concave in $\mu$, and setting its derivative to zero gives
\[
27(1-b-\mu)(2a-\mu)=64\mu(\mu-a),\quad\text{i.e.}\quad37\mu^2+(27(1-b)-10a)\mu-54a(1-b)=0.
\]
Its root $\mu^*\in(a,2a)$ is the maximizer, and $V(1/2)=\frac12\bigl(V_{\sfL}(\mu^*,\theta^*)+V_{\sfR}(\mu^*,\theta^*)\bigr)$.
\end{itemize}

\paragraph{Choosing $(l,r)$.}
By Corollary~\ref{cor:half-stars-K}, we need a line $\tau l+(1-\tau)r$ lying above $\min\{\widehat V(\tau),\tau,(1-\tau)/2\}$ on $[0,1]$. Let $\tau_-$ be the point where the left piece of $\widehat V$ meets $\tau$, and $\tau_+$ the point where the right piece meets $(1-\tau)/2$:
\[
\tau_-\coloneqq\frac{V(0)}{1+2V(0)-2V(1/2)},\qquad\tau_+\coloneqq1-\frac{V(1)}{1/2-2V(1/2)+2V(1)}.
\]
Take the line through $(\tau_-,\tau_-)$ and $(\tau_+,(1-\tau_+)/2)$, that is,
\[
l\coloneqq\frac{(1-\tau_+)(1-3\tau_-)}{2(\tau_+-\tau_-)},\qquad r\coloneqq\frac{\tau_-(3\tau_+-1)}{2(\tau_+-\tau_-)}.
\]
On $[\tau_-,\tau_+]$ it is a chord of the convex function $\widehat V$, so it lies above $\widehat V$. Its slope $l-r$ lies in $(-1/2,1)$, so it lies above $\tau$ to the left of $\tau_-$ and above $(1-\tau)/2$ to the right of $\tau_+$. Hence $(l,r)\in\calK(1/2)$, and by Corollary~\ref{cor:circuit-to-pm}, $(l,r)$ is an admissible exponent pair. Figure~\ref{fig:numerics} illustrates these functions.

\begin{figure}[H]
\centering
\pgfplotsset{
  bwplot/.style={
    width=\linewidth, height=5.6cm, axis lines=left, xlabel={$\tau$},
    tick label style={font=\scriptsize}, label style={font=\small},
    every axis plot/.append style={line join=round},
  },
  envelope/.style={line width=3.2pt, black!22},
  vhat/.style={black, thin, mark=*, mark size=1.1pt},
  diag/.style={black, densely dashed},
  anti/.style={black, dash dot},
  secant/.style={black, semithick},
}
\begin{subfigure}[t]{0.48\textwidth}
\centering
\begin{tikzpicture}
\begin{axis}[bwplot, xmin=0, xmax=1, ymin=0, ymax=0.55,
  xtick={0,0.25,0.5,0.75,1}, ytick={0,0.1,0.2,0.3,0.4,0.5},
  legend columns=5, legend to name=numericslegend,
  legend style={font=\scriptsize, draw=none, /tikz/every even column/.append style={column sep=1.2em}}]
\addplot[envelope] coordinates {(0,0) (0.331683,0.331683) (0.5,0.249979) (0.501373,0.249313) (1,0)};
\addlegendentry{$\min\{\widehat V(\tau),\tau,\frac{1-\tau}2\}$}
\addplot[vhat] coordinates {(0,0.492686) (0.5,0.249979) (1,0.007578)};
\addlegendentry{$\widehat V(\tau)$}
\addplot[diag] coordinates {(0,0) (0.55,0.55)};
\addlegendentry{$\tau$}
\addplot[anti] coordinates {(0,0.5) (1,0)};
\addlegendentry{$\frac{1-\tau}2$}
\addplot[secant] coordinates {(0,0.492684) (1,0.007276)};
\addlegendentry{$\tau l+(1-\tau)r$}
\end{axis}
\end{tikzpicture}
\caption{Overview on $[0,1]$.}\label{fig:numerics-overview}
\end{subfigure}\hfill
\begin{subfigure}[t]{0.48\textwidth}
\centering
\begin{tikzpicture}
\begin{axis}[bwplot, xmin=0.496, xmax=0.5024, ymin=-1.2, ymax=0.6,
  xtick={0.496,0.498,0.5,0.502}, extra x ticks={0.501373},
  extra x tick labels={$\tau_+$},
  extra x tick style={tick label style={font=\scriptsize, yshift=-2.4ex}},
  ytick={-1,-0.5,0,0.5},
  scaled x ticks=false, xticklabel style={/pgf/number format/fixed, /pgf/number format/precision=3},
  ylabel={difference ($\times10^{-6}$)}, ylabel style={yshift=-0.6ex}]
\addplot[envelope] coordinates {(0.496,-0.81130) (0.5,-0.83105) (0.5013730,0) (0.5014553,-1.2)};
\addplot[vhat, mark indices={2}] coordinates {(0.496,-0.81130) (0.5,-0.83105) (0.5023643,0.6)};
\addplot[anti] coordinates {(0.5013319,0.6) (0.5014553,-1.2)};
\addplot[secant] coordinates {(0.496,0) (0.5024,0)};
\end{axis}
\end{tikzpicture}
\caption{Zoom near $\tau=1/2$, shifted.}\label{fig:numerics-zoom}
\end{subfigure}

\medskip
\pgfplotslegendfromname{numericslegend}
\caption{The choice of $(l,r)$. In~(b), every curve is shifted by $-(\tau l+(1-\tau)r)$, so the requirement is that the envelope stays below $0$.}\label{fig:numerics}
\end{figure}

\paragraph{Conclusion.}
Put $\xi\coloneqq1/4-(l+r)/2$, so that $l+r=1/2-2\xi$, and apply Theorem~\ref{thm:subset-sum-transfer} with $\varepsilon\coloneqq10^{-9}$. By the table,
\[
\frac12-\frac{\xi}{9}+10^{-9}=0.49999777486376221650\ldots<\frac12-\frac1{450000}=0.49999\overline{7},
\]
which proves Theorem~\ref{thm:main}.

\begin{table}[H]
\centering
\renewcommand{\arraystretch}{1.15}
\begin{tabular}{@{}lr@{}}
\toprule
Quantity & Value \\
\midrule
$\kappa$ & $0.00533321875848117877\ldots$ \\
$4a$ & $0.00151819462943814496\ldots$ \\
$V(0)$ & $0.49268569663735450584\ldots$ \\
$V(1/2)$ & $0.24997913372377961141\ldots$ \\
$V(1)$ & $0.00757766963945240565\ldots$ \\
$\mu^*$ & $0.00069649211514613618\ldots$ \\
$\theta^*$ & $0.00092865615352818158\ldots$ \\
$\tau_-$ & $0.33168260605143320534\ldots$ \\
$\tau_+$ & $0.50137304583751125366\ldots$ \\
$l$ & $0.00727587055962990941\ldots$ \\
$r$ & $0.49268405898808998771\ldots$ \\
$(l+r)/2$ & $0.24997996477385994856\ldots$ \\
$\xi$ & $0.00002003522614005143\ldots$ \\
$1/2-\xi/9$ & $0.49999777386376221650\ldots$ \\
\bottomrule
\end{tabular}
\caption{Numerical values at the chosen parameters.}\label{tab:numerics}
\end{table}

\FloatBarrier
\phantomsection\section*{AI disclosure and timestamps}\label{sec:ai-disclosure}

In July and August 2026, we were developing a general theory of depth-2 linear circuits, building on~\cite{JS13,AL25,Ye26one}. In August, while looking for applications of this theory, we found reductions from generalized subset sum (GSS) to compatibility testing for tensor powers of matrices, through the representation technique of~\cite{NW21,RW26}. In particular, we knew that GSS with $C=[\pm2]\coloneqq\{-2,-1,1,2\}$ could beat meet-in-the-middle if some circuit for the partial match matrix $P_{n,1/2}$ achieved $l+r<1/2$.

On 5 September, we used GPT-6 Pro to construct the first circuit in our project meeting this criterion, a $\SUPZ$ circuit. From 5 to 28 September, we worked to incorporate this new GSS result into the existing theory. More generally, the reduction works for any fixed $s\in(0,1/2]$ if some circuit for $P_{n,s}$ achieves $l+r<\h(s)-s$, where $\h(s)-s$ is the corresponding $\OR$ benchmark; this allowed us to shift from $s=1/2$ to $s=1/8$.

On 28 September, we used GPT-6 Astra to find a reduction showing that ordinary subset sum can also beat meet-in-the-middle using the same circuit for $P_{n,s}$. This gave our improved algorithm for subset sum; Record~1 arrives at the candidate exponent $0.499664500303\ldots$.

We then pursued the simplest self-contained proof, at a substantial cost in the exponent. For this, we shifted back to $s=1/2$ and used an $\EXACTZ$ circuit instead. The manuscript locked on \textbf{4 October at 02:49} gives a randomized algorithm in worst-case time $O(2^{(1/2-1/450000)n})$, and Record~6 reports a computational reproduction at \textbf{03:51} that morning.

The table below lists, in chronological order, six records of our conversations with ChatGPT during the development of this work. Their timestamps are taken from the original message metadata, in Beijing time (UTC+8). These records were made public on \textbf{7 October 2026}.

\begin{longtable}[]{@{}
  >{\raggedright\arraybackslash}p{(\columnwidth - 4\tabcolsep) * \real{0.12}}
  >{\raggedright\arraybackslash}p{(\columnwidth - 4\tabcolsep) * \real{0.42}}
  >{\raggedright\arraybackslash}p{(\columnwidth - 4\tabcolsep) * \real{0.46}}@{}}
\toprule\noalign{}
\begin{minipage}[b]{\linewidth}\raggedright
Record
\end{minipage} & \begin{minipage}[b]{\linewidth}\raggedright
Link
\end{minipage} & \begin{minipage}[b]{\linewidth}\raggedright
Original message range, 2026
\end{minipage} \\
\midrule\noalign{}
\endhead
\bottomrule\noalign{}
\endlastfoot
1 &
\href{https://chatgpt.com/share/6ac5a437-5d0c-83e8-afea-fcc450b322b0}{Computation
Trees Criterion} & Sep 16, 14:34--Sep 28, 06:48 \\
2 &
\href{https://chatgpt.com/share/6ac5a295-7bb4-83ec-85cf-b0034f63d081}{View
recent files} & Sep 28, 07:16--07:31 \\
3 &
\href{https://chatgpt.com/share/6ac5a378-a208-83ee-9457-4c3ce74b64a4}{View
project conversations} & Sep 28, 18:03--Sep 30, 23:39 \\
4 &
\href{https://chatgpt.com/share/6ac5a3f1-d070-83e8-9474-754bd863b090}{Recover
the improved P1 construction} & Oct 1, 03:46--Oct 4, 02:49 \\
5 &
\href{https://chatgpt.com/share/6ac5a537-691c-83ec-b103-6d7423e7f97f}{Conceptual
simplification} & Oct 3, 05:57--19:07 \\
6 &
\href{https://chatgpt.com/share/6ac5a5e1-8828-83ee-99c7-0d470c4bfc48}{Reproduce
correctness check} & Oct 4, 03:46--03:51 \\
\end{longtable}

A preliminary version of this work was subsequently posted on arXiv~\cite{YL26v1} on 7 October 2026.

\bibliographystyle{alphaurl}
\bibliography{main}
\end{document}